\documentclass[journal,10pt]{IEEEtran}

\usepackage{cite}
\usepackage{color} 
\usepackage{stfloats}

\makeatletter
\renewcommand\normalsize{%
	\@setfontsize\normalsize\@xpt\@xiipt
	\abovedisplayskip 1.2\p@ \@plus1\p@ \@minus5\p@
	\abovedisplayshortskip \z@ \@plus1\p@
	\belowdisplayshortskip 4\p@ \@plus2\p@ \@minus2\p@
	\belowdisplayskip \abovedisplayskip
	\let\@listi\@listI}
\makeatother

\ifCLASSINFOpdf
   \usepackage[pdftex]{graphicx}
   
\else
\fi
\usepackage{cite}
\usepackage{amsmath,amsthm,amssymb,graphicx,algorithm,algorithmic}
\usepackage{makecell}
\usepackage{subfigure}
\usepackage{float}
\newtheorem{theorem}{\bf Theorem}

\usepackage{color}   
\usepackage{array}
\usepackage{multicol}
\begin{document}

\title{Reconfigurable Intelligent Surface Assisted High-Speed Train Communications: Coverage Performance Analysis and Placement Optimization}

\author{

Changzhu~Liu,~\IEEEmembership{Graduate~Student~Member,~IEEE},
Ruisi~He,~\IEEEmembership{Senior~Member,~IEEE}, \\
Yong Niu,~\IEEEmembership{Senior Member,~IEEE},
Zhu~Han,~\IEEEmembership{Fellow,~IEEE}, 
Bo~Ai,~\IEEEmembership{Fellow,~IEEE}, 
Meilin~Gao,~\IEEEmembership{Member,~IEEE}, \\ 
Zhangfeng~Ma,~\IEEEmembership{Member,~IEEE}, 
Gongpu~Wang,~\IEEEmembership{Member,~IEEE},
and Zhangdui~Zhong,~\IEEEmembership{Fellow,~IEEE}

\thanks{
Copyright (c) 2015 IEEE. Personal use of this material is permitted. However, permission to use this material for any other purposes must be obtained from the IEEE by sending a request to pubs-permissions@ieee.org.
This work is supported by National Key R\&D Program of China under Grant 2020YFB1806903, the National Natural Science Foundation of China under Grant 62271037 and 62001519,  the Fundamental Research Funds for the Central Universities under Grant 2022JBQY004, the State Key Laboratory of Advanced Rail Autonomous Operation under Grant RCS2022ZZ004, the Natural Science Foundation of Hunan Province under Grant 2023JJ40607, the Beijing Engineering Research Center of High-Speed Railway Broadband Mobile Communications under Grant BHRC-2022-3, the NSF CNS-2107216, CNS-2128368, CMMI-2222810, ECCS-2302469, US Department of Transportation, Toyota and Amazon. Part of this work  has been presented at the IEEE Wireless Communications and Networking Conference (WCNC) 2023. \emph{(Corresponding author:Ruisi He; Yong Niu.) }

C. Liu, R. He, and Z. Zhong are with the State Key Laboratory of Advanced Rail Autonomous Operation, the School of Electronics and Information Engineering, the Frontiers Science Center for Smart High-speed Railway System, the Beijing Engineering Research Center of High-Speed Railway Broadband Mobile Communications, Beijing Jiaotong University, Beijing 100044, China (e-mails: changzhu{\_}liu@bjtu.edu.cn; ruisi.he@bjtu.edu.cn; zhdzhong@bjtu.edu.cn).

Y. Niu is with the State Key Laboratory of Advanced Rail Autonomous Operation, Beijing Jiaotong University, Beijing 100044, China, and also with the National Mobile Communications Research Laboratory, Southeast University, Nanjing 211189, China (e-mail: niuy11@163.com).

Z. Han is with the Department of Electrical and Computer Engineering at the University of Houston, Houston, TX 77004 USA, and also with the Department of Computer Science and Engineering, Kyung Hee University, Seoul, South Korea, 446-701 (e-mail: hanzhu22@gmail.com).

B. Ai is with the State Key Laboratory of Advanced Rail Autonomous Operation, Beijing Jiaotong University, Beijing 100044, China, and also with Henan Joint International Research Laboratory of Intelligent Networking and Data Analysis, Zhengzhou University, Zhengzhou 450001, China, and also with Research Center of Networks and Communications, Peng Cheng Laboratory, Shenzhen, China (e-mail: boai@bjtu.edu.cn).

M. Gao is with the Tsinghua Space Center and the Beijing National Research Center for Information Science and Technology, Tsinghua University, Beijing 100091, China (email: gaomeilin@tsinghua.edu.cn).

Z. Ma is with the College of Information Engineering, Shaoyang University, Shaoyang 422000, China (e-mail: zhangfeng.ma@vip.126.com)

G. Wang is with the Beijing Key Lab of Transportation Data Analysis and Mining, School of Computer and Information Technology, Beijing Jiaotong University, Beijing 100044, China (e-mail: gpwang@bjtu.edu.cn).

              }
\thanks{
              }}

\markboth{IEEE Transactions on Vehicular Technology,~Vol.~XX, No.~XX, XXX~2022}
{}

\maketitle

\begin{abstract} %
Reconfigurable intelligent surface (RIS) emerges as an efficient and promising technology for the next wireless generation networks and has attracted a lot of attention owing to the capability of extending wireless coverage by reflecting signals toward targeted receivers. In this paper, we consider a RIS-assisted high-speed train (HST) communication system to enhance wireless coverage  and improve coverage probability. First, coverage performance of the downlink single-input-single-output system is investigated, and the closed-form expression of coverage probability is derived. Moreover, travel distance maximization problem is formulated to facilitate RIS discrete phase design and RIS placement optimization, which is subject to coverage probability constraint. Simulation results validate that better coverage performance and higher travel distance can be achieved with deployment of RIS. The impacts of some key system parameters including transmission power, signal-to-noise ratio threshold, number of RIS elements, number of RIS quantization bits, horizontal distance between base station and RIS, and speed of HST on system performance are investigated. In addition, it is found that RIS can well improve coverage probability with limited power consumption for HST communications.
\end{abstract}
\begin{IEEEkeywords}
Reconfigurable intelligent surface (RIS), high-speed train communication, coverage probability, travel distance.
\end{IEEEkeywords}

%
\IEEEpeerreviewmaketitle
\vspace{-4mm}
\section{Introduction}
\IEEEPARstart{H}{igh-speed} train (HST) communications have attracted a lot of attention in recent years, and it has tended to evolve from informatization to intelligence with integration of mobile communication and artificial intelligence technologies\cite{re2,re1,Huang1,Huang2,5GR}. The future HST communications will support more intelligent applications and this puts forward higher requirements for system performance. Compared with the traditional wireless communication networks, HST communications have high mobility of onboard transceivers and large signal penetration loss through train cars, which lead to many challenges, such as channel modeling, coverage enhancement, Doppler shift compensation, time-varying channel estimation, beamforming design, and resource management \cite{re3,He,ch1,ch2,chen}. Some key technologies of the fifth-generation (5G) and the sixth generation (6G) communication systems, such as massive multi-input-multi-output (MIMO), millimeter-wave (mmWave), renon-orthogonal multiple access, unmanned aerial vehicle (UAV), and reconfigurable intelligent surface (RIS) will gradually be integrated with HST communications to improve system performance \cite{6g1,ma1,ma2,ma3,re7}.

HST communications are vast different from public terrestrial cellular communication, especially in terms of reliability and safety \cite{re21}. In the design of HST communication systems, cell coverage performance is an important indicator, which is critical to guarantee reliable data transmission. Up to now, there are only few literatures working on cell coverage and improvement for HST communications.  In \cite{re30}, percentage of cell coverage area and edge outage probability were evaluated. In order to achieve seamless coverage for environment-diverse HST, a space-ground integrated cloud railway network was proposed in \cite{re8}, which can reduce handover time and improve coverage. In \cite{re31} and \cite{re32}, a beamforming based coverage performance improvement scheme was proposed for HST communication systems with elliptical cells. In \cite{re33} and \cite{re34}, by considering overlap area between adjacent cells and  hard handoff scheme, coverage performance of HST system was analyzed. In order to reduce impact of handoff, \cite{re36} proposed two different free-space-optics coverage models to improve coverage performance. In \cite{re37} and \cite{re38}, UAV was introduced to assisted HST communications for coverage improvement. In \cite{re39}, coverage analysis and performance optimization for HST communication systems with carrier aggregation were investigated, and theoretical expressions for edge coverage probability and percentage of cell coverage area were derived. The cell-free massive MIMO HST communication system with multiple  antenna points and multiple users was investigated in \cite{re399}, and  deterministic signal-to-interference-plus-noise ratio expression and  tight upper bound on coverage probability were derived. Among them, coverage performance can be improved, however, energy consumption, manufacturing and deployment cost of base station (BS) cannot be effectively reduced. On the other hand, HST wireless propagation environment is not intelligently controlled, and the range of establishing effective communication between BS and HST is still limited, which further limits the improvement effect of the coverage performance. 

RIS is one of the promising technologies for future wireless communication. It is capable of smartly designing radio environment, and has the characteristics of low cost, low complexity, and easy deployment \cite{ref2,ref3,sun,RIS1,RIS2,RIS3}, which can be used in future high-mobility communications (such as vehicular communications, UAV communications, HST communications et al.). In recent years, there has been some works on the investigation of RIS-assisted high-mobility communications,  which are in the areas of channel estimation \cite{ce1,ce2}, transmission protocol \cite{tp1,tp2}, Doppler effect mitigation \cite{dem1,dem2,dem3,dem4} and beamforming \cite{beam1}. In addition, investigation of RIS-assisted HST communications is still at its infancy. In \cite{re3}, the authors provided a RIS-aided HST wireless communication paradigm, including its main challenges and application scenarios, and provided solution to signal processing and resource management. In \cite{re55} and \cite{re552}, the authors considered RIS-assisted MIMO downlink system for HST communications, where outage probability with statistic channel state information (CSI) was investigated. Multiple RISs assisted HST communication system was investigated in \cite{re56}, and the authors considered a train-ground time division duplexing communication paradigm to deploy two RISs, and further solved spectrum effective maximization problem.The RIS-assisted mmWave HST communications were studied in \cite{re57}, and  spectral efficiency maximization problem was formulated, and performance was evaluated through deep strengthening learning. In \cite{re58}, the authors considered RIS-assisted free-space-optics communications for HST access connectivity, where average signal-to-noise ratio (SNR) and outage probability were analyzed. In \cite{re59} and \cite{re60}, the authors investigated interference suppression for RIS-assisted railway communication systems. However, wireless coverage performance of HST communications is not investigated among the existing works.

Wireless coverage in complex environment can be significantly enhanced with the aid of RISs \cite{re40,re41}, and it is thus helpful for HST coverage enhancement. For coverage analysis of RIS-assisted communications, there are some works in recent years \cite{re43,re44,re45,re47,re48,re54,re51,re46,noma4}. In \cite{re43}, a RIS-assisted network model was investigated, and a RIS placement optimization problem was formulated to maximize cell coverage. In \cite{re44}, the authors investigated a RIS-assisted point-to-point network without direct link, and analyzed network coverage, SNR gain, and delay outage rate. A RIS-assisted single-input single-output (SISO) system was considered in \cite{re45}, where coverage probability and ergodic capacity were analyzed. A RIS-assisted terahertz wireless systems was considered in \cite{re47}, and a novel exact closed-form expression of coverage probability was derived. In \cite{re48}, the authors investigated coverage performance of RIS-assisted wireless networks without direct link for the Nakagami-$m$ channel. A double-RIS assisted communication system was considered in \cite{re54}, and simulation results validated the enhanced performance over single-RIS counterpart. In \cite{re51}, a RIS-assisted SISO system was considered and the closed-form expression for coverage probability was derived. In \cite{re46}, the authors investigated RIS-assisted mmWave cellular networks, where peak reflection power expression and downlink signal-to-interference ratio coverage probability expression were derived. RIS-aided multi-cell networks were studied in \cite{noma4}, and the expression for coverage probability and ergodic rate were derived and performance was analyzed. Although RIS can be used to enhance coverage of wireless communications, however, impact of RIS placement on coverage performance has not been well investigated among \cite{re44,re45,re46,re47,re48,noma4,re51,re54}. There are some existing works that concentrate on optimizing a single RIS placement. In \cite{place1}, the authors investigated spectrum sharing scenarios between radars and wireless communication systems, and via jointly optimizing the placement of BS and RIS to maximize probability of coverage of RIS-assisted communications. The placement of a single RIS was analyzed in three-dimensional environments in \cite{place2}. In \cite{place4}, the authors investigated optimal placement of one RIS given fixed locations of one transmitter and one receiver in a mmWave link, in order to maximize end-to-end SNR. The placement optimization of multi-RISs was investigated in \cite{place5,place6,place7}. In more detail, in \cite{place5}, a minimization problem of number of RISs was formulated by jointly optimizing  number, locations, and phase shift coefficients of RISs. In \cite{place6} and \cite{place7}, RIS deployment problem was studied for RIS-aided downlink multi-user communication systems, and  performance gain was increased by optimizing RIS locations. In addition, to the authors' best knowledge, investigation of RIS-assisted HST communication coverage is still missing.

Motivated by the above gap, coverage performance analysis and RIS placement optimization for RIS-assisted SISO downlink HST communications are investigated in this paper. In our previous work of \cite{mywork}, we briefly introduced the idea of using RIS to enhance coverage probability, and compared coverage probability with the case without RIS. In this paper, we extend our previous work and consider impact of RIS placement on HST communication  performance. Another three schemes are considered for comparison, i.e., ideal phase, random phase and without RIS. The major contributions can be summarized as follows: 
\begin{itemize}
\item We consider a novel RIS-assisted SISO downlink HST communication system to extend coverage of HST, where RIS is deployed to provide reflective paths to enhance received power at mobile relay (MR).
\item The closed-form expression of coverage probability is derived according to CSI for RIS-assisted HST communications. Then, travel distance maximization problem is formulated subject to coverage probability constraint. Based on the alternative optimization method, a joint RIS discrete phase design and RIS placement optimization algorithm is proposed.
\item Simulations are presented which demonstrates the impact of some key system parameters on coverage probability, travel distance, and transmission rate of HST communications, such as transmission power, SNR threshold, number of RIS elements, number of RIS quantization bits, horizontal distance between BS and RIS, and speed of HST. Different schemes are compared and it is demonstrated that the deployment of RIS can effectively improve coverage performance of HST communication systems.
\end{itemize}

The remainder of this paper is organized as follows. System model is described in Section II. In Section III, coverage probability of RIS-assisted HST communication systems are derived and travel distance maximization problem is formulated. The solution of travel distance maximization problem is presented in Section IV. In Section V, numerical results are presented to show the impacts of some key system parameters on coverage performance. Finally, conclusions are given in Section VI.
\vspace{-3mm}
\section{System Modle}
\subsection{Scenario Description}
A RIS-assisted HST communication system model is proposed in this section, as shown in F{}ig.~\ref{fig:1a}. We consider a single-antenna BS serves a single-antenna MR with the help of RIS. Trackside BS transmits signals to a train which is bypassing the coverage cell, and a MR is mounted on top of the train to avoiding penetration loss. RISs are deployed a position close to the track within cell coverage. With the help of RIS, a reflective channel between BS and MR is established to improve signal quality. Therefore, for RIS-assisted HST communications, MR can not only receive the direct signals from BS, but also receive the reflected signals via RIS.

For practical application, the time of RIS phase shift operation $T_{\rm{RIS}}$  depends on the electromagnetic pulse according to the regulation voltage, which is in orders of ns \cite{new1}. To measure the time-variant effect in HST communications, the stationarity time is introduced during which time the channel keeps constant or has no great change. $T_{\rm{slot}}$  is in orders of milliseconds in HST mmWave communications \cite{new2}. Thus, it is valid and practical for the almost real-time RIS operation in time-varying HST communications

As show in \cite{t11}, there are two cases for path-loss model in RIS-aided communication, i.e., near field and far field. Specifically, when distance between the BS/MR and the center of RIS is less than $\frac{2L^2}{\lambda}$, where $L$ and $\lambda$ represent the largest size of RIS and the wave length of signal, respectively, BS/MR are considered to be in the near field of the RIS. Otherwise, they are classified as the far field of RIS. In this paper, we consider the far field case\footnote{In this paper, we consider from vertical BS-to-rail distance $d_{\rm{BS}}^{\rm{v}}=50$ m and vertical  from RIS-to-rail  distance $d_{\rm{RIS}}^{\rm{v}}=20$ m, as shown in F{}ig.~\ref{fig:1b}. In other words, vertical distance from BS to the center of RIS and vertical distance from the center of RIS to MR are approximately equal $30$ m and $20$ m, respectively, and we can calculate $L$ approximately equal $1.13$ m or $1.38$ m according to $\frac{2L^2}{\lambda} = 20$ m or $30$ m when carrier frequency $f=2.35$ GHz. Let $U$ denote the side of a square RIS element. When $L=1.13$ m, if $U=\frac{\lambda}{10}$ or $U=\frac{\lambda}{2}$, $N$ can exceed $7800$ or $300$, respectively, while $L=1.38$ m, $N$ can exceed $11000$ of $400$, respectively. However, we consider smaller number of RIS elements in this paper, namely, distance between BS/MR and the center of RIS is larger than $\frac{2L^2}{\lambda}$. Thus, BS/MR are classified as in the far field of the RIS \cite{t11}.}.

\begin{figure*}[!t]
  \centering 
  \subfigure[]
  {
    {\includegraphics[scale=0.085]{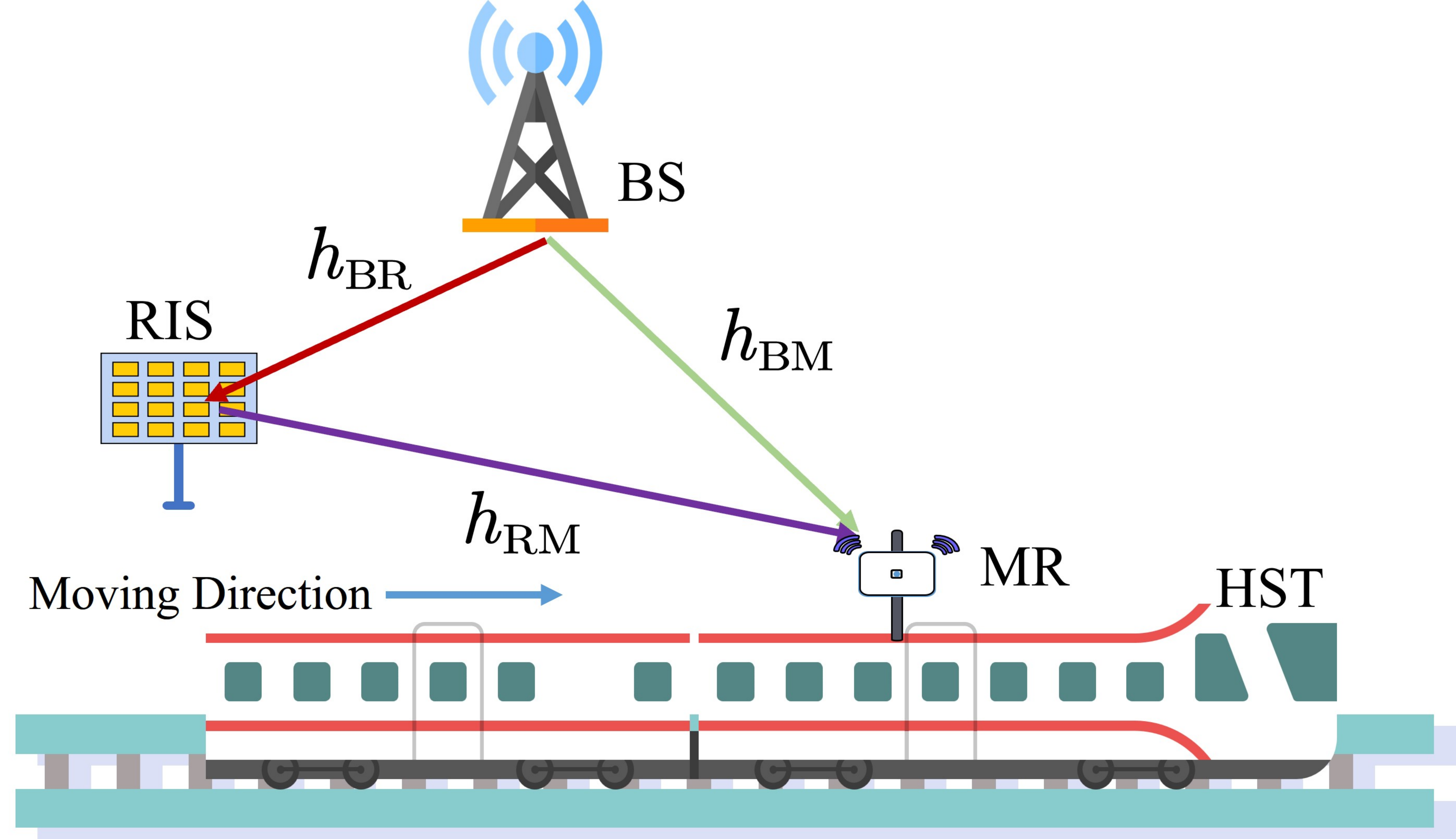}\label{fig:1a}} 
  }
  \subfigure[]
  {
    {\includegraphics[scale=0.18]{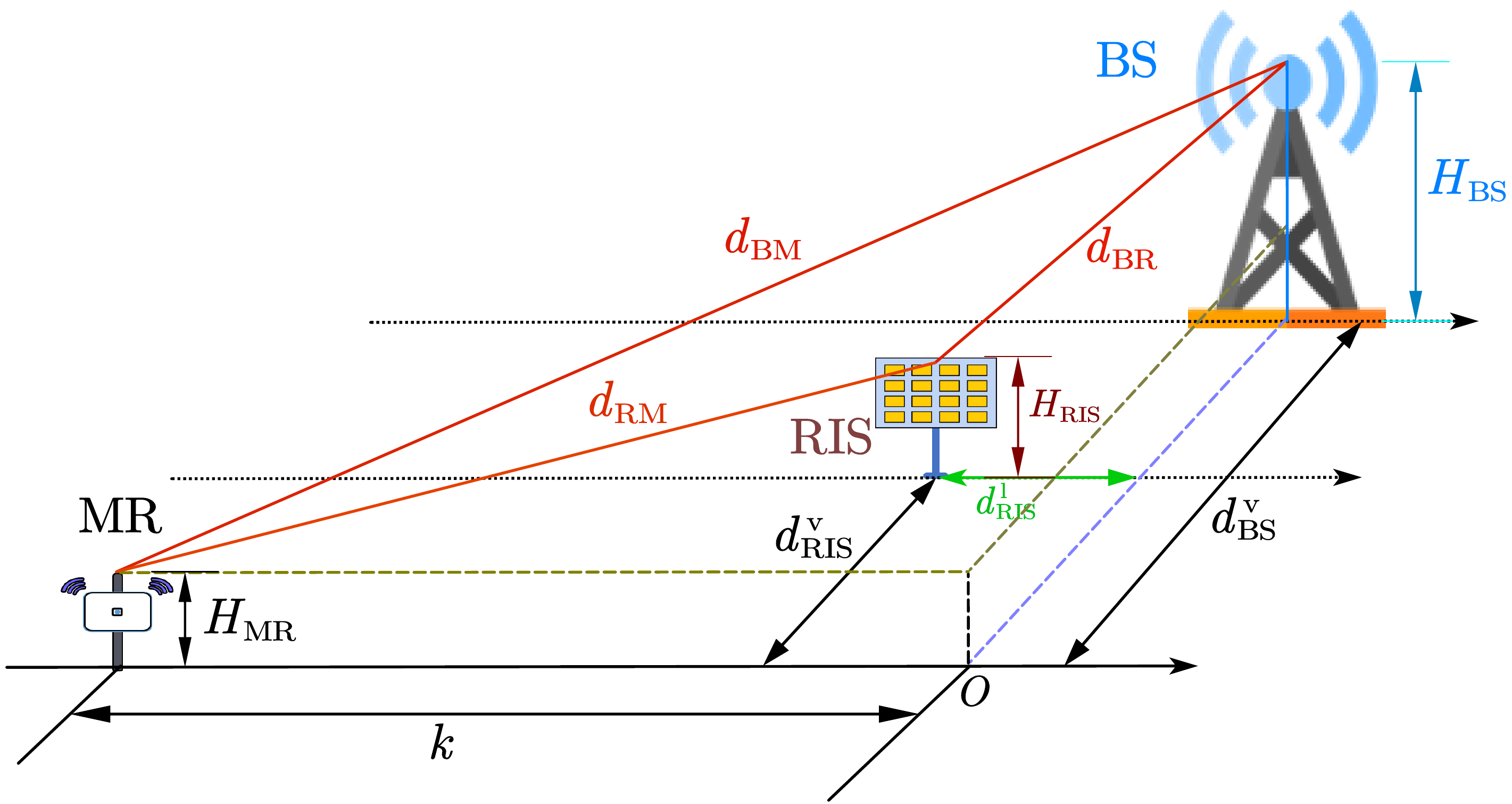}\label{fig:1b}} 
  }
  \caption{ \label{fig:1} (a): Illustration of a RIS-assisted SISO downlink system for HST communications.  (b): Geometrical representation of  RIS-assisted HST communications.}
\end{figure*}

\vspace{-3mm}
\subsection{Channel Model}
In the system, a RIS is equipped with $N$ reflecting elements. We assume that each element of  RIS is capable of independently rescattering signal, which can be dynamically adjusted by RIS controller \cite{t2}, and assume that HST firstly moves with speed $v$  close to BS from a certain position $k$ on the left of BS, and then moves far away from BS.  Let $R$ denote coverage radius of BS, and total time $T$ slots are considered, $\tau$ is the slot duration, and $\mathcal{N} = \left\{1,\cdots, N \right\}$. Being widely used in HST communications, a Rician fading channel is considered for HST scenario \cite{t3,t33}. All the links follow the Rician fading since both line-of-sight (LoS)  and non-line-of-sight (NLoS) exist \cite{t4}. Note that, the transceiver distances (BS-MR, and RIS-MR links) always change across time slots due to high mobility of HST. The Doppler shift caused by train moving affects the signal receiving. However, for HST wireless transmission scheme, the determined train direction and repetitive movement of trains along the fixed tracks causes the repeated and predictable Doppler shift effects, and it is thus possible to track and compensate the Doppler effect \cite{n1,n2}. In addition, there have some research results on the Doppler effect estimation and compensation for HST communications \cite{n3,n4,n5,n6}. Therefore, it is reasonable to assume that the Doppler effect can be well compensated. The detailed descriptions of BS-MR, BS-RIS and RIS-MR channels are given in the following.

\subsubsection{BS-MR Channel}
The channel between BS and  MR is denoted as $h_{\rm{BM} }\left( t \right) \in \mathbb{C}$. In time slot $t \in \left\{1,\cdots, T \right\}$, it can be expressed as
\vspace{1.5mm}
\begin{equation} \label{eq:hbm}
  h_{\rm{BM}} \left( t \right)=\sqrt{\frac{\kappa _{\rm{BM}}}{\kappa _{ \rm{ BM} }+1}}\bar{h}_{ \rm{BM}}\left( t \right)+\sqrt{\frac{1}{\kappa _{ \rm{BM} }+1}}\tilde{h}_{ \rm{BM}}\left( t \right),
\end{equation}
where $\kappa_{\rm{BM}} \geq 0 $ is the Rician $K$-factor, $\bar{h}_{\rm{BM}}\left( t \right)\in \mathbb{C} $ is the LoS component depending on BS-MR link and remains stable with in each time slot\footnote{We ignore shadowing effects and assume that large-scale component is determined only by distance-based path-loss \cite{t5}.}, and $\tilde{h}_{ \rm{BM}}\left( t \right) \in \mathbb{C} $ is the NLoS component. According to \cite{t6}, the LoS component of the  channel between  BS and MR can be given as 
\begin{equation}
  \bar{h}_{\rm{BM}}\left( t \right)=\sqrt{PL _{\rm{BM}}\left( t \right)}e^{-j\theta_{\rm{BM}}\left( t \right)},
\end{equation}
 where $PL_{\rm{BM}}\left( t \right) = d_{\rm{BM}}^{-\varepsilon_{\rm{BM}}}\left( t \right) > 0$  represents the distance-dependent path-loss, and $d_{\rm{BM}}\left( t \right)$ is the distance between BS and MR, as shown in F{}ig.~\ref{fig:1b}, which can be calculated as
\begin{equation}
  d_{\rm{BM}}\left( t \right) = \sqrt{\left( d_{\rm{BS}}^{\rm{v}} \right) ^2+\left( H_{\rm{BS}}-H_{\rm{MR}} \right) ^2+\left( k-vt\tau \right) ^2},
\end{equation}
where $d_{\rm{BS}}^{\rm{v}}$ denotes vertical distance from  BS to rail, $H_{\rm{BS}}$ and $H_{\rm{MR}}$ denote the heights of the BS and MR, respectively, and $k$ denotes initial horizontal distance of HST close to the BS, and $\varepsilon_{\rm{BM}}$ is  path-loss exponent, $\theta _{\rm{BM}}\left( t \right) = \frac{2\pi}{\lambda}d_{\rm{BM}}\left( t \right)$ is phase, where $\lambda = \frac{c}{f}$ is wavelength, $c$ denotes light speed, and $f$ represents carrier frequency. Similarly, the NLoS component can be written as
\begin{equation}
  \tilde{h}_{ \rm{BM}}\left( t \right) = \sqrt{PL _{\rm{NLoS}}^{\rm{BM}}\left( t \right)}\tilde{h}_{\rm{NLoS}}^{\rm{BM}}\left( t \right),
\end{equation}
where  $PL _{\rm{NLoS}}^{\rm{BM}}\left( t \right) = d_{\rm{BM}}^{-\varepsilon'_{\rm{BM}}}\left( t \right) > 0$ is the distance-dependent path-loss for the NLoS case, $\varepsilon'_{\rm{BM}}$ is path-loss exponent, and $\tilde{h}_{\rm{NLoS}}^{\rm{BM}}\left( t \right)\sim \mathcal{C} \mathcal{N} (0,1)$ denotes small-scale channel component.

\subsubsection{RIS-Assisted Channel}
$\forall n \in \mathcal{N}$, we assume that the channel between BS and the $n$-th RIS element, and between the $n$-th RIS element and MR are denoted as $h_{\rm{BR}}^{n}\left( t \right) \in \mathbb{C}$ and $h_{\rm{RM}}^{n}\left( t \right) \in \mathbb{C}$, respectively. In time slot $t$, $h_{\rm{BR}}^{n}\left( t \right)$ and $h_{\rm{RM}}^{n}\left( t \right)$ can be given as
\begin{equation} \label{eq:hbr}
  h_{\rm{BR}}^{n} \left( t \right)=\sqrt{\frac{\kappa _{\rm{BR}}}{\kappa _{ \rm{ BR} }+1}}\bar{h}_{ \rm{BR}}^{n}+\sqrt{\frac{1}{\kappa _{ \rm{BR} }+1}}\tilde{h}_{ \rm{BR}}^{n}\left( t \right),
\end{equation}
\begin{equation} \label{eq:hrm}
  h_{\rm{RM}}^{n} \left( t \right)=\sqrt{\frac{\kappa _{\rm{RM}}}{\kappa _{ \rm{RM}}+1}}\bar{h}_{ \rm{RM}}^{n}\left( t \right)+\sqrt{\frac{1}{\kappa _{ \rm{RM} }+1}}\tilde{h}_{ \rm{RM}}^{n}\left( t \right),
\end{equation}
where $\kappa_{\rm{BR}} \geq 0 $ and $\kappa_{\rm{RM}} \geq 0 $ are Rician $K$-factors, $\bar{h}_{ \rm{BR}}^{n}\in \mathbb{C}$ and  $ \bar{h}_{ \rm{RM}}^{n}\left( t \right) \in \mathbb{C}$ denote the LoS components, and $\tilde{h}_{ \rm{BR}}^{n}\left( t \right)\in \mathbb{C}$ and $\tilde{h}_{ \rm{RM}}^{n}\left( t \right) \in \mathbb{C}$ denote the NLoS components. Similar to the BS-MR link, we have 
\begin{equation}
  \bar{h}_{ \rm{BR}}^{n} =  \sqrt{PL _{\rm{BR}}^{n}}e^{-j\theta_{\rm{BR}}^{n}}, 
\end{equation}
\begin{equation}
  \bar{h}_{ \rm{RM}}^{n}\left( t \right) =  \sqrt{PL _{\rm{RM}}^{n}\left( t \right)}e^{-j\theta_{\rm{RM}}^{n}\left( t \right)},
\end{equation}
 where $PL_{\rm{BR}}^{n} = \left(d_{\rm{BR}}^{n}\right )^{-\varepsilon_{\rm{BR}}}>0$, $PL_{\rm{RM}}^{n}\left( t \right) = \left(d_{\rm{RM}}^{n} \left( t \right)\right )^{-\varepsilon_{\rm{RM}}} > 0$ represent  path-loss exponents, where $d_{\rm{BR}}^{n}$ and $d_{\rm{RM}}^{n}\left( t \right)$ denote the distance between BS and the $n$-th RIS and the distance between the $n$-th RIS element and MR, respectively, as shown in F{}ig.~\ref{fig:1b}, which can be given as
\begin{equation}
  d_{\rm{BR}}^{n} = \sqrt{\left( d_{\rm{RIS}}^{\rm{l}} \right) ^2+\left( H_{\rm{BS}}-H_{\rm{RIS}} \right) ^2+\left( d_{\rm{BS}}^{\rm{v}}-d_{\rm{RIS}}^{\rm{v}} \right) ^2},
\end{equation}
\begin{equation}
  d_{\rm{RM}}^{n}\left( t \right) = \sqrt{\left( d_{\rm{RIS}}^{\rm{v}} \right) ^2+\left( H_{\rm{RIS}}-H_{\rm{MR}} \right) ^2+\left( k-vt\tau -d_{\rm{RIS}}^{\rm{l}} \right) ^2},
\end{equation}
 where $d_{\rm{RIS}}^{\rm{l}}$ denotes horizontal distance between  BS and  RIS. Since  BS and  MR are far away from RIS, their distances from different RIS elements are approximately the same, i.e., $d_{\rm{BR}}^{n}\left( t \right) = d_{\rm{BR}}^{\rm{c}}\left( t \right)$, and $d_{\rm{RM}}^{n}\left( t \right) = d_{\rm{RM}}^{\rm{c}}\left( t \right)$, where $d_{\rm{BR}}^{\rm{c}}\left( t \right)$ and $d_{\rm{RM}}^{\rm{c}}\left( t \right)$ are the distance between BS and the center of RIS and the distance between MR and the center of RIS, respectively, and $d_{\rm{RIS}}^{\rm{v}}$ denotes vertical distance from RIS to rail, $H_{\rm{RIS}}$ denotes the height of the RIS, and $\varepsilon_{\rm{BR}}$, $\varepsilon_{\rm{RM}}$ denote path-loss exponents, $\theta_{\rm{BR}}^{n}\ = \frac{2\pi}{\lambda}d_{\rm{BR}}^{n}$, $\theta_{\rm{RM}}^{n}\left( t \right) = \frac{2\pi}{\lambda}d_{\rm{RM}}^{n}\left( t \right)$ are phases. Similarly, the NLoS components can be written as 
 \vspace{1mm}
 \begin{equation}
\tilde{h}_{ \rm{BR}}^{n}\left( t \right) = \sqrt{ PL _{\rm{NLoS,BR}}^{n}}\tilde{h}_{\rm{NLoS,BR}}^{n}\left( t \right),
 \end{equation}
 \begin{equation}
  \tilde{h}_{ \rm{RM}}^{n}\left( t \right) = \sqrt{PL _{\rm{NLoS,RM}}^{n}\left( t \right)}\tilde{h}_{\rm{NLoS,RM}}^{n}\left( t \right),
 \end{equation}
 where $PL _{\rm{NLoS,BR}}^{n} = \left(d_{\rm{BR}}^{n}\right) ^{-\varepsilon'_{\rm{BR}}}> 0$, $ PL _{\rm{NLoS,RM}}^{n}\left( t \right) = \left(d_{\rm{RM}}^{n}\left( t \right)\right)^{-\varepsilon'_{\rm{RM}}} > 0$ denote the distance-dependent path-loss for NLoS case, $\varepsilon'_{\rm{BR}}$, $\varepsilon'_{\rm{RM}}$ are path-loss exponents, $\tilde{h}_{\rm{NLoS,BR}}^{n}\left( t \right), \tilde{h}_{\rm{NLoS,RM}}^{n}\left( t \right)\sim \mathcal{C} \mathcal{N} (0,1)$ denote small-scale channel components. 

According to the aforementioned, the received signal at MR in time slot $t$ can be expressed as
\begin{align}
    y\left( t \right) &= \sqrt{P}\left(  h_{\rm{BM}} \left( t \right) + \sum_{n=1}^N{h^{n}_{{\rm{RM}} }\left( t \right) e^{j\theta _n\left( t \right)}h^{n}_{{\rm{BR}} }}\left( t \right)   \right) x\left( t \right) \\
    &+ z\left( t \right), \nonumber 
\end{align}
where $P$ denotes transmission power, $x(t)\sim \mathcal{C} \mathcal{N} (0,1)$ denotes the signal transmitted to MR during time slot $t \in \left\{1,\cdots, T \right\}$ with zero mean and variance equal $1$,  $z\left( t \right) \sim \mathcal{C} \mathcal{N} (0,\sigma ^2)$ denotes i.i.d. additive white Gaussian noise at MR with zero mean and variance $\sigma^{2} = 0$, and $\theta _n\left( t \right)$ represents phase of the $n$-th RIS element. For ease of  implementation, we consider that the phase at each RIS element can only take a finite discrete values with equal quantization intervals $\left[0, 2\pi \right)$. Let $b$ denote number of RIS quantization bits. Then the set of phases at each element is given by $ \mathcal{F} =\left\{ 0, \Delta\theta ,\cdots ,\Delta\theta  \left( M-1 \right) \right\} $, where $ \Delta\theta  =\frac{2\pi}{M}$ and $M=2^b$. Accordingly, the equivalent channel between BS and MR can be represented as
\begin{equation} \label{eq:h}
  h\left( t \right) = h_{\rm{BM}} \left( t \right) + \sum_{n=1}^N{h^{n}_{{\rm{RM}} }\left( t \right) e^{j\theta _n\left( t \right)}h^{n}_{{\rm{BR}} }}\left( t \right),
\end{equation}
\vspace{3mm}
and the SNR at MR in time slot $t$ is given by
\vspace{4mm}
\begin{equation}
  \gamma \left( t \right) = \frac{P\left|h_{\rm{BM}} \left( t \right) + \sum_{n=1}^N{h^{n}_{{\rm{RM}} }\left( t \right) e^{j\theta _n\left( t \right)}h^{n}_{{\rm{BR}} }}\left( t \right)  \right|^2 }{\sigma^2},
\end{equation}
where $\left| \cdot \right|$ is the absolute value of a complex value.
 
\vspace{-3mm}
\section{ Coverage Probability and Travel Distance Maximization}
In this section, at first,  we derive the expression of coverage probability according to CSI. Second, we aim to maximize travel distance by jointly considering RIS discrete phase $\theta _n\left( t \right)$ optimization and horizontal distance between BS and RIS $ d_{\rm{RIS}}^{\rm{l}}$ optimization. 
\vspace{-3mm}
\subsection{Coverage Probability}
Coverage probability $P_{\rm{cov}}\left( t \right)$ is defined as the probability that the effectively received SNR $\gamma \left( t \right)$ at MR is larger than a given SNR threshold $\gamma_{th}$, which can be given by
\vspace{1mm}
\begin{equation} \label{eq:pcov}
  \begin{array}{l}
    P_{\rm{cov}}\left( t \right)=\mathrm{Pr}\left( \gamma \left( t \right) \geqslant \gamma _{\rm{th}} \right) \\
  =\mathrm{Pr}\left( \frac{P\left| h_{\rm{BM}} \left( t \right) + \sum_{n=1}^N{h^{n}_{{\rm{RM}} }\left( t \right) e^{j\theta _n\left( t \right)}h^{n}_{{\rm{BR}} }}\left( t \right) \right|^2}{\sigma ^2}\geqslant \gamma _{\rm{th}} \right) \\
  =1-\mathrm{Pr}\left( \frac{P\left| h_{\rm{BM}} \left( t \right) + \sum_{n=1}^N{h^{n}_{{\rm{RM}} }\left( t \right) e^{j\theta _n\left( t \right)}h^{n}_{{\rm{BR}} }}\left( t \right) \right|^2}{\sigma ^2}<\gamma _{\rm{th}} \right) \\ 
  =1-\mathrm{Pr} \left( \left| h_{\rm{BM}} \left( t \right) + \sum_{n=1}^N{h^{n}_{{\rm{RM}} }\left( t \right) e^{j\theta _n\left( t \right)}h^{n}_{{\rm{BR}} }}\left( t \right) \right|^2    \right. \\ \left.
  \quad <\frac{\gamma _{\rm{th}}}{\bar{\gamma}}  \right) \\
  =1-\mathrm{Pr}\left( \left| h\left( t \right) \right|^2<\frac{\gamma _{\rm{th}}}{\bar{\gamma}} \right), \\
  \end{array}
\end{equation}
where $\bar{\gamma}  = \frac{P}{\sigma^2}$ denotes the average transmission SNR, and let $ P_{\rm{out}}\left( t \right) = \mathrm{Pr}\left( \left| h\left( t \right) \right|^2<\frac{\gamma _{\rm{th}}}{\bar{\gamma}} \right)$, thus $P_{\rm{cov}}\left( t \right)$ can be written as 
\begin{equation}
P_{\rm{cov}}\left( t \right) = 1 -  P_{\rm{out}}\left( t \right).
\end{equation}

The probability distribution of $h\left( t \right)$ is given by Theorem~\ref{theorem1}.

\begin{theorem}
  \label{theorem1}
  The equivalent channel $h\left(t\right)$ between BS and MR, follows complex-valued Gaussian distribution with mean $\mu_h\left(t\right)$ and  variance $\sigma_{h}^2\left(t\right)$, namely, $h\left(t\right) \sim \mathcal{C} \mathcal{N} \left( \mu_h\left(t\right),\sigma_{h}^2\left(t\right)  \right)$, and we have
\begin{align}
    \mu_h\left(t\right) &=\rho _{\rm{BM} }\sqrt{PL _{\rm{BM} }\left(t\right)}e^{-j\theta _{ \rm{BM}}\left(t\right)} \\ \nonumber
    &+\sum_{n=1}^N\rho_{\rm{RM} }\rho _{\rm{BR} }\sqrt{PL _{ \rm{RM}}^{n}\left(t\right)}\sqrt{PL _{ \rm{BR}}^{n}} \\ \nonumber
    &\times e^{j\left( \theta _n\left(t\right)-\theta_{\rm{RM}}^{n}\left(t\right)-\theta_{\rm{BR}}^{n} \right)},\nonumber
\end{align}
\vspace{1mm}
\begin{align}
  \sigma _{h}^{2}\left(t\right)&=\varrho_{\rm{BM}}^2PL _{\rm{NLoS}}^{\rm{BM}}\left(t\right) \\ \nonumber
  &+\sum_{n=1}^{N}\varrho_{\rm{RM}}^2\varrho_{\rm{BR}}^2PL _{\rm{NLoS,RM}}^{n}\left(t\right)PL _{{\rm{NLoS,BR}}}^{n}\left(t\right), \nonumber
\end{align}
where $\rho_{\mathrm{BM}}=\sqrt{\frac{\kappa _{{\mathrm{BM}}}}{\kappa _{{\mathrm{BM}}}+1}}$, $\varrho _{{\mathrm{BM}}}=\sqrt{\frac{1}{\kappa _{\mathrm{BM}}+1}}$, $\rho _{\mathrm{BR}}=\sqrt{\frac{\kappa _{{\mathrm{BR}}}}{\kappa _{{\mathrm{BR}}}+1}}$, $\varrho _{{\mathrm{BR}}}=\sqrt{\frac{1}{\kappa _{{\mathrm{BR}}}+1}}$, $\rho _{\mathrm{RM}}=\sqrt{\frac{\kappa _{\mathrm{RM}}}{\kappa _{\mathrm{RM}}+1}}$, and $\varrho _{{\mathrm{RM}}}=\sqrt{\frac{1}{\kappa _{\mathrm{RM}}+1}}$.
\end{theorem}

\begin{proof}
See in Appendix A.
\end{proof}

According to Theorem~\ref{theorem1}, coverage probability can be derived as in Theorem~\ref{theorem2}.
\begin{theorem}
  \label{theorem2}
  $P_{\rm{out}}\left( t \right)$ follows a non-central chi-square distribution, i.e., $\chi ^2 \left(\nu,\zeta \left( t \right)\right)$, with the degree of freedom $\nu = 2 $, and the non-centrality parameter $\zeta\left( t \right) = \frac{\mu_h^2\left(t\right)}{\sigma_{h}^2\left(t\right)}$, where
\begin{equation}
    \begin{array}{l}\left| \mu_h\left( t \right) \right|^2= \left|\rho _{\mathrm{BM}}\sqrt{PL _{\mathrm{BM}}\left( {t} \right)}e^{-j\theta _{_{\mathrm{BM}}}\left( {t} \right)} \right. \\ \left.
    +\sum_{n=1}^N{\rho _{\mathrm{RM}}\rho _{\mathrm{BR}}\sqrt{PL _{_{\mathrm{RM}}}^{n}\left( t \right)}\sqrt{PL _{_{\mathrm{BR}}}^{n}}e^{j\left( \theta _n\left( {t} \right) -\theta _{\mathrm{RM}}^{n}\left( {t} \right) -\theta _{\mathrm{BR}}^{n} \right)}}\right|^2.
\end{array}
\end{equation}

With the corresponding cumulative distribution function (CDF), $P_{\rm{out}}\left( t \right)$ is given by
\begin{equation}
  P_{\rm{out}}\left( t \right) = 1 - Q_{1}\left(\sqrt{\zeta\left( t \right)}, \sqrt{\gamma_0\left( t \right)}\right),
\end{equation}
where $\gamma_0 = \frac{\gamma_{th}}{\bar{\gamma}\sigma _{h}^{2}\left( t \right)}$, and  $Q_m \left(a,b \right)$ is the Marcum Q-function defined in \cite{t8}. Thus, coverage probability can be expressed as 
\begin{align}
  P_{\rm{cov}} = 1 - P_{\rm{out}} = Q_{1}\left(\sqrt{\zeta\left( t \right)}, \sqrt{\gamma_0\left( t \right)}\right).
\end{align}
\end{theorem}

\begin{proof}
  See Appendix B.
\end{proof}
\vspace{-1mm}
In order to verify of Theorem 2, Fi{}g.~{\ref{fig:2}} depicts the analysis and simulation results of coverage probability $P_{\rm{cov}}$ against time slot, and compares it with the case that ideal phase, random phase, andwithout considering RIS deployment. As aforementioned in Section II, HST firstly moves with speed $v$ close to BS from a certain position $k$ on the left of BS, and then moves far away from BS, and RISs are deployed a position close to the track. The main settings of the simulation are shown in Table I. This figure shows fairly well match between analysis and simulation results. It is observed that coverage probability firstly increases with time slot. When time slot is further increased, coverage probabilities converge to the maximum value and remains for a period of time. Finally, with time slot gradually increasing to the end, coverage probabilities under the four schemes decrease. The reason is that channel state is gradually improved when HST moves close to BS from distance $k$ initially, thus coverage probability increases. Afterward, the channel remains a fairly good stage for a period of time, whereas coverage probability is fixed to the maximum value. Finally, coverage probability decreases gradually since channel state becomes worse when HST moves far away from BS. In addition, we observe that coverage probability of the proposed algorithm is higher than the cases without RIS and with random phase, and has similar performance to ideal phase case. In a word, it is clear that when the phases of RIS elements are optimized, and RIS-assisted system has better coverage performance than the case that without RIS deployment.
\begin{figure}[!t]
  \centering 
 {\includegraphics[scale=0.6]{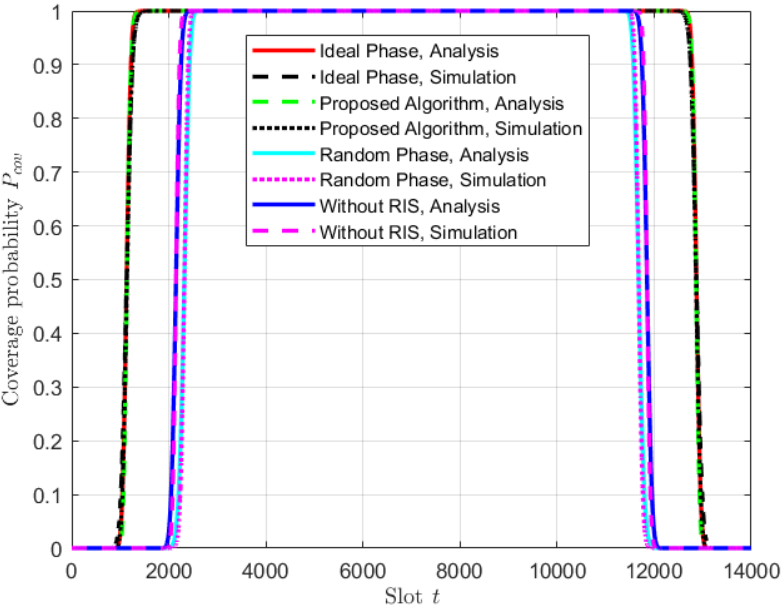}} 
  \caption{ \label{fig:2} Coverage probability vs. slot. System parameters: $k = 7,000$ m, $T=14,000$, $P = 30$ dBm, $\gamma_{th} = 10$ dB, $N=60$ and $b=2$.} 
\end{figure}

\vspace{-3mm}
\subsection{Travel Distance Maximization } 
In this paper, we target on HST traveled distance $D $ maximization by jointly considering RIS elements discrete phase $\theta _n\left( t \right)$ optimization and horizontal distance between BS and RIS $d_{\rm{RIS}}^{\rm{l}}$ optimization, namely, RIS placement optimization, subject to constraint of coverage probability, where $D$ is defined as the maximum travel distance when coverage probability satisfies a given coverage probability threshold in total time $T$ slots. We consider discrete values of $d_{\rm{RIS}}^{\rm{l}}$, as shown in F{}ig.~\ref{fig:1b}, and let $\mathcal{D} = \left\{0,\Delta d,\cdots,\Delta d\left(S-1\right) \right\} $, where $\Delta d = \frac{R}{S}$, $S$ denotes the number of interval. 
The HST traveled distance maximization problem within each time slot is mathematically formulated as
\vspace{-3mm}
\begin{align}  \label{eq:P1} 
    \mathcal{P}_1: \underset{d_{\rm{RIS}}^{\rm{l}}, \theta _n\left( t \right)}{\max } & D =\sum_{t=1}^T\beta\left(t\right) {d_t}  \\ 
    \mathrm{s}.\mathrm{t}. \quad & P_{\rm{cov}}\left( t \right) \geq P_{\rm{th}} \tag{\ref{eq:P1}a},   \label{yya} \\ 
    &\theta _n\left( t \right) \in \mathcal{F} ,n\in N,  \tag{\ref{eq:P1}b} \\ 
    & d_{\rm{RIS}}^{\rm{l}} \in \mathcal{D},   \tag{\ref{eq:P1}c} \\
    & \beta\left(t\right) \in \left\{0,1 \right\}, \tag{\ref{eq:P1}d} 
  \end{align}
where $\beta\left(t\right) \in \left\{0,1 \right\}$ is a binary variable, i.e., $\beta\left(t\right) = 1$, if $P_{\rm{cov}}\left( t \right) \geq P_{\rm{th}}$, otherwise, $\beta\left(t\right) = 0$.  Parameter $d_t = v\tau$ denotes the traveled distance of HST in a slot duration, and $P_{\rm{th}}$ is coverage probability threshold.

\vspace{-3mm}
\section{Solution to the Travel Distance Maximization Problem}
In this section, we design a maximal travel distance algorithm to solve problem $\mathcal{P}_1$ by jointly RIS discrete phase and placement optimization based on the alternating optimization. The key idea is to optimize one of the variables while fixing another variable in an alternating method, according to the decoupling of the two variables. Hence, we may transform the optimization problem formulated in $\mathcal{P}_1$ into two subproblems, i.e., discrete phase optimization subproblem and RIS placement subproblem, and solve the two subproblems successively. Finally, we can obtain the optimal strategy $\left( d_{\rm{RIS}}^{l,\ast}, \theta _n^\ast\left( t \right)\right)$, where $ d_{\rm{RIS}}^{l, \ast} $ and $\theta _n^\ast\left( t \right)$ denote the optimal horizontal distance between BS and RIS and RIS phase, respectively.
\vspace{-3mm}
\subsection{Discrete Phase Optimization Subproblem}
In this subsection, we suppose that RIS is fixed in a position, i,e,. $d_{\rm{RIS}}^{\rm{l}}$ is given, and the optimization problem formulated in $\mathcal{P}_1$ can be expressed as
\vspace{-3mm}
\begin{align} \label{eq:P2}
  \mathcal{P}_2: \underset{\theta _n\left( t \right)}{\max } & \ D =\sum_{t=1}^T\beta\left(t\right) {d_t}   \\  
     \mathrm{s}.\mathrm{t}. \ & P_{\rm{cov}}\left( t \right) \geq P_{\rm{th}}, \tag{\ref{eq:P2}a} \\  
     &\theta _n\left( t \right) \in \mathcal{F} ,n\in N,  \tag{\ref{eq:P2}b} \\
     & \beta\left(t\right) \in \left\{0,1 \right\}. \tag{\ref{eq:P2}c}
\end{align}

\vspace{-1.5mm}
According to the expression of coverage probability \eqref{eq:pcov}, it is related to channel response. However, the RIS phase has an effect on the channel gain of $h(t)$, which also has  influence on coverage probability $P_{\rm{cov}}\left( t \right)$. Since, $\beta\left(t\right) = 1$, $d_t = v\tau$, if $P_{\rm{cov}}\left( t \right) \geq P_{\rm{th}}$, otherwise, $\beta\left(t\right) = 0$, i.e., travel distance equal zero at time slot $t$. In other words, we should make coverage probability as large as possible to satisfy coverage probability constraint \eqref{yya}. Thus, $\mathcal{P}_2$ can be equivalent to maximize coverage probability  $P_{\rm{cov}}\left( t \right)$  as follows:
\begin{align} \label{eq:P22}
  \mathcal{P}_3: \underset{\theta _n\left( t \right)}{\max } & \ P_{\rm{cov}}\left( t \right)   \\  
     \mathrm{s}.\mathrm{t}. \ & \theta _n\left( t \right) \in \mathcal{F} ,n\in N. \tag{\ref{eq:P22}a}  
\end{align}

When fixing the horizontal distance between BS and RIS $d_{\rm{RIS}}^{\rm{l}}$, we can optimize RIS phase. Parameter $\mathcal{F}$ contains a series of discrete variables, and the range available for each phase depends on RIS quantization bits. Considering the complexity and validity, we exploit the local search method as shown in Algorithm~\ref{phase} to solve this problem. Specifically, keeping the other $N-1$ phase values fixed, for each element $\theta_n\left( t \right)$, we traverse all possible values and choose the optimal one without violating coverage probability constraint. Then, we use this optimal solution $\theta _n^\ast\left( t \right)$ as the new value of $\theta _n\left( t \right)$  for the optimization of another phase, until all phases in the set $\mathcal{F}$ are fully optimized.
\begin{algorithm}[!t]
  \caption{Local Search for Discrete Phase}
  \label{phase}
  \begin{algorithmic}[1]
  \REQUIRE
 number of RIS elements $N$, number of quantization bits $b$
  \ENSURE
    {$\theta _n^\ast\left( t \right),~\forall n \in \mathcal{N}$ }
  \FOR {$n=1:N$}
  \FOR {$\theta _n\left( t \right)\in \mathcal{F}$}
  \STATE Assign all possible values to $\theta _n\left( t \right)$, and select the value maximizing coverage probability $P_{\rm{cov}} \left( t \right)$ denoted as $\theta _n^\ast\left( t \right) =\left\{\arg{\max }\ P_{\rm{cov}}\left( t \right) \left|\theta _n\left( t \right)\in \mathcal{F} \right.\right\}$;
  \STATE $\theta _n\left( t \right)=\theta _n^\ast\left( t \right)$;
  \ENDFOR
  \ENDFOR
  \end{algorithmic} 
\end{algorithm} 

Note that, in this paper, for a time slot $t$, even if each variable in $\mathcal{F}$  is traversed, coverage probability constraint may not be satisfied due to poor channel state. In addition, $h\left( t \right)$ varies with the location of HST, and it is impossible to satisfy coverage probability constraint for all time slot $t \in \left\{1,\cdots,T \right\}$ even with the optimal phase $\theta _n^\ast\left( t \right)$. Therefore, coverage probability is not only related to RIS phase $\theta_n\left( t \right)$, but also to the equivalent channel $h\left( t \right)$ between BS and MR.

\vspace{-4mm}
\subsection{RIS Placement Subproblem}
In this subsection, we suppose the RIS elements phases are given, i.e., $\theta _n\left( t \right)$ is given, and the optimization problem formulated in $\mathcal{P}_1$ can be expressed as
\begin{align} \label{eq:P3}
  \mathcal{P}_2: \underset{d_{\rm{RIS}}^{\rm{l}}}{\max } & \ D =\sum_{t=1}^T\beta\left(t\right) {d_t} \\  
     \mathrm{s}.\mathrm{t}. \ & P_{\rm{cov}}\left( t \right) \geq P_{\rm{th}},  \tag{\ref{eq:P3}a}   \\
& d_{\rm{RIS}}^{\rm{l}} \in \mathcal{D}, \tag{\ref{eq:P3}b} \\
& \beta\left(t\right) \in \left\{0,1 \right\}. \tag{\ref{eq:P3}c}
\end{align}

Similar to the discrete phase optimization subproblem, we exploit the local search method as shown in Algorithm~\ref{rispos} to solve this problem. Specifically, keeping all RIS elements phase values fixed, for $d_{\rm{RIS}}^{\rm{l}}$, we traverse all possible values in set $\mathcal{D}$ and judge each time slot whether violating coverage probability constraint until HST runs through all time slots. Finally, we can obtain the optimal RIS position denoted by $ d_{\rm{RIS}}^{l,\ast} = \left\{\arg{\max }\ D \left| d_{\rm{RIS}}^{\rm{l}} \in  \mathcal{D} \right.\right\}$.

\begin{algorithm}[!t]
  \caption{Local Search for RIS Placement}
  \label{rispos}
  \begin{algorithmic}[1]
  \REQUIRE
    the total time $T$ slots, RIS elements phases $\theta _n\left( t \right),~\forall n \in \mathcal{N} $
  \ENSURE
    {$d_{\rm{RIS}}^{l,\ast}$}
  \FOR {$d_{\rm{RIS}}^{\rm{l}} \in \mathcal{D}$}
  \FOR {$t=1:T$}
  \STATE Calculate $D$ defined in \eqref{eq:P1} on the premise that constraint $P_{\rm{cov}}\left( t \right) \geq P_{\rm{th}}$ is satisfied; 
  \ENDFOR
  \ENDFOR
  \STATE $d_{\rm{RIS}}^{l,\ast} = \left\{\arg{\max }\ D \left| d_{\rm{RIS}}^{\rm{l}} \in \mathcal{D} \right.\right\}$;
  \end{algorithmic}
\end{algorithm}

\begin{algorithm}[!t]
  \caption{Travel Distance Maximization}
  \label{traveldis}
  \begin{algorithmic}[1]
  \REQUIRE
    the total time $T$ slots,  number of RIS elements $N$, number of quantization bits $b$, random generate RIS elements phases $\theta _n\left( t \right),~\forall n \in \mathcal{N} $
  \FOR {$d_{\rm{RIS}}^{\rm{l}} \in \mathcal{D}$}
  \FOR {$t=1:T$}
  \FOR {$n=1:N$}
  \FOR {$\theta _n\left( t \right)\in \mathcal{F}$}
  \STATE Assign all possible values to $\theta _n\left( t \right)$, and select the value maximizing coverage probability $P_{\rm{cov}} \left( t \right)$ denoted as $\theta _n^\ast\left( t \right)=\left\{\arg{\max }\ P_{\rm{cov}}\left( t \right) \left|\theta _n\left( t \right)\in \mathcal{F} \right.\right\}$;
  \STATE   $\theta _n\left( t \right)=\theta _n^\ast\left( t \right)$;
  \ENDFOR
  \ENDFOR
  \STATE Calculate $D$ defined in \eqref{eq:P1}  on the premise that constraint  $\eqref{yya}$ is satisfied 
  \ENDFOR
  \ENDFOR
  \STATE $\left( d_{\rm{RIS}}^{l,\ast}, \theta _n^\ast\left( t \right)\right)= \left\{ \arg{\max }\ D \left| d_{\rm{RIS}}^{\rm{l}} \in \mathcal{D} \right., \theta _n\left( t \right) \in \mathcal{F} \right\}$
  \STATE  $D_{\max} = \left\{ D=\sum_{t=1}^T\beta\left(t\right){d_t}\left| d_{\rm{RIS}}^{l,\ast}, \theta _n^\ast\left( t \right) \right.  \right\}$
  \end{algorithmic}
\end{algorithm}

\vspace{-3mm}
\subsection{Travel Distance Maximization}
We summarize the above discrete phase optimization subproblem algorithm and RIS placement subproblem algorithm, and propose the travel distance maximization algorithm. For the given optimal $d_{\rm{RIS}}^{l,\ast}$ and $\theta _n^\ast\left( t \right)$, the optimization problem formulated in $\mathcal{P}_1$ can be expressed as
\begin{align} \label{eq:P4}
  \mathcal{P}_3: \underset{d_{\rm{RIS}}^{l,\ast},\theta _n^\ast\left( t \right)}{\max } & \ D =\sum_{t=1}^T\beta\left(t\right) {d_t} \\  
     \mathrm{s}.\mathrm{t}. \ \ & P_{\rm{cov}}\left( t \right) \geq P_{\rm{th}}, \tag{\ref{eq:P4}a} \\
     & \beta\left(t\right) \in \left\{0,1 \right\}. \tag{\ref{eq:P4}b}
\end{align}

As shown in Algorithm~\ref{traveldis}, we set MR and RIS at a given position in initial state, and randomly generate a initial phase value for $\theta _n\left( t \right),~n \in \mathcal{N}$. Then, we update the RIS placement position and phase in an alternating manner until traverse all possible values in set $\mathcal{D}$. Finally, we can obtain the optimal strategy as
\begin{equation}
  \left( d_{\rm{RIS}}^{l,\ast}, \theta _n^\ast\left( t \right)\right)= \left\{ \arg{\max}~D\left| d_{\rm{RIS}}^{\rm{l}}\in \mathcal{D}, \theta _n\left( t \right) \in \mathcal{F}, n \in \mathcal{N} \right. \right\},
\end{equation}
 and the maximum travel distance is denoted by
 \begin{equation}
  D_{\max} = \left\{ D=\sum_{t=1}^T\beta\left(t\right){d_t}\left| d_{\rm{RIS}}^{l,\ast}, \theta _n^\ast\left( t \right) \right.  \right\}.
 \end{equation}

\section{Numerical Analysis}
In this section, we analyze coverage performance of RIS-assisted HST communication system with the proposed algorithms. Simulation results are provided to validate system performance in terms of coverage probability, travel distance, and transmission rate. The simulation parameters are set as listed in Table~\ref{tab:1}. For comparison, another three schemes are considered:
\begin{itemize}
  \item {\bf Ideal Phase:} this scheme considers that the ideal continuous phase is $\theta_{n}^{\ast}\left(t\right) =-\theta _{\rm{BM}}\left(t\right)-\left( \theta_{\rm{RM}}^{n}\left(t\right)+\theta_{\rm{BR}}^{n} \right)$, $\theta_{n}^{\ast}\left(t\right) \in \left[0, 2\pi \right)$, $n \in \mathcal{N} $, which yields a received signal at MR having the largest amplitude \cite{re43,re45}. 
  \item {\bf Random Phase:} this scheme randomly selects a value of phase $\theta _n\left( t \right)$ for each RIS element and keeps unchanged, where $\theta_n\left( t \right)$ follows an independent uniform distribution over $\left[0, 2\pi \right)$.
  \item {\bf Without RIS:} this scheme does not use RIS for signal reflection, and MR can only receive signals from BS. 
\end{itemize} 

\begin{table*}[!t]   
  \caption{Simulation Parameters}
  \label{tab:1}  
  \centering 
  \begin{tabular}{|c|c|c|c|c|c|} \hline
    \textbf{Parameter}  & \textbf{Symbol}    & \textbf{Value} & \textbf{Parameter}  & \textbf{Symbol}    & \textbf{Value}     \\
    \hline
     {Speed of HST}        & $v$   & $360$ km/h & Noise power   & $\sigma^2$    & \makecell{ $-174~ {\rm{dBm/Hz}}$ \\ $+10\log _{10}B+N_f$}  \\  
    \hline
    \ {Height of BS}       & $H_{\rm{BS}}$      &  $10$ m   & Noise figure & $N_{f}$  &  $10$ dB \\
    \hline
     {Height of RIS}       & $H_{\rm{RIS}}$     &  $2$ m  & Carrier frequency  & $f$  & $2.35$ GHz \\
    \hline
     {Height of MR}        & $H_{\rm{MR}}$      &  $2.5$ m  & {Coverage probability threshold} & $P_{\rm{cov}}$  & $95\%$ \\
    \hline
     Distance from BS to rail  & $d_{\rm{BS}}^{\rm{v}}$  & $50$ m &  Rician $K$-factor  & $\kappa_{\left( {\rm{BM,RM,BR}}\right)}$ & $10$ dB \\
    \hline
    {Distance from RIS to rail} & $d_{\rm{RIS}}^{\rm{v}}$  &   $20$ m & NLoS case path-loss exponent  &  $\varepsilon'_{\left( {\rm{BM,RM,BR}}\right)}$  & $3.6$  \\ 
    \hline
    {System bandwidth}    & $B$     & $20$ MHz & LoS case path-loss exponent  & $\varepsilon_{\left( {\rm{BM,RM,BR}}\right)}$ & $3$  \\
    \hline
  \end{tabular}
  
\end{table*}

\vspace{-6mm}
\subsection{Coverage Probability}
In this subsection, the impact of the four schemes on coverage probability with different parameters is analyzed. Fi{}g.~{\ref{fig:3}} illustrates coverage probability against transmission power $P$ under different number of RIS elements. It can be observed that coverage probabilities under the four schemes increase with increasing transmission power. When transmission power is further increased, coverage probabilities converge to the maximum value. We can also see that the proposed algorithm outperforms the cases without RIS and random phase except for the case ideal phase. In addition, we observe that coverage probability of the proposed algorithm increases with increasing number of RIS elements. This is because when $N$ increases, channel gain of $h\left(t\right)$ increases, which results that the received power at MR is enhanced and coverage probability is thus improved. It is further observed that the proposed algorithm requires lower transmission power than the case without using RIS for the fixed coverage probability value, and the required transmission power decreases with increasing number of RIS elements. This implies that RIS deployment can reduce power consumption and manufacturing cost of transceiver.
\begin{figure}[!t]
  \centering
  {\includegraphics[scale=0.6]{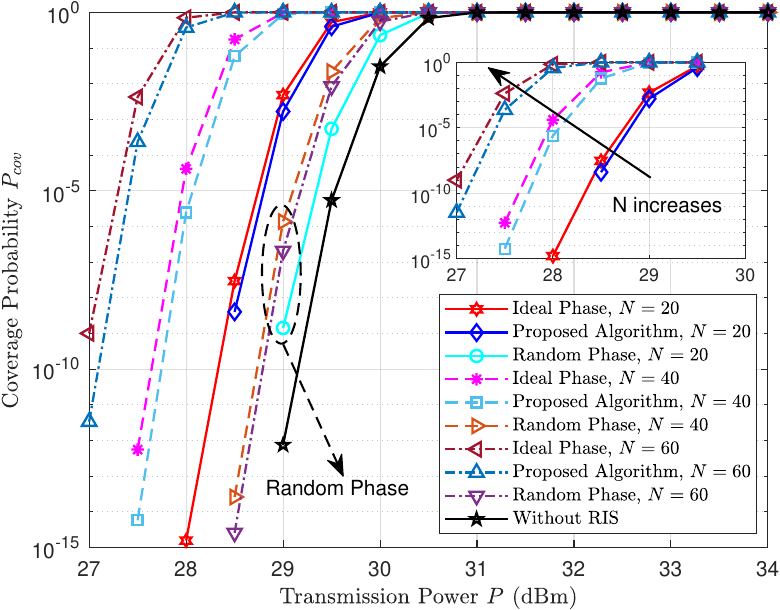}}
  \caption{ \label{fig:3}Coverage probability vs. transmission power. System parameters: $k = 11,000$ m, $T=22,000$, $\gamma_{th} = 10$ dB, and $b=2$.}
\end{figure}

Fi{}g.~{\ref{fig:4}} plots coverage probability versus SNR threshold $\gamma_{th}$ under different number of quantization bit $b$. It can be observed that coverage probabilities under the four schemes decrease with increasing of SNR threshold. When $\gamma_{th}$ increases to infinity, coverage probability of all schemes are $0$. Obviously, improving SNR threshold requirements means more communication outage events and lower coverage probability. We also observe that coverage probability of the proposed algorithm increases with $b$. This is because phase resolution increases with $b$, which results that the received power at MR is enhanced, and leads to  improvement of coverage probability. As expected, the performances of the case ideal phase and random phase are not affected by the number of RIS quantization bits. In addition, it is found that coverage probability of $b = 1$ is significantly lower than $b = 3 $ and $b = 5$ at higher $\gamma_{th}$, which means that coverage probability is significantly improved when $b > 1$. It can be further observed that narrow coverage probability gap exists between $b = 3$ and $b = 5$ and between ideal phase and $b=3,5$, which indicates that coverage probability may remain unchanged with further increasing $b$, and coverage probability obtained by the proposed algorithm approaches the results with ideal phase as number of quantization bits increases.
\begin{figure}[!t]
  \centering
  {\includegraphics[scale=0.6]{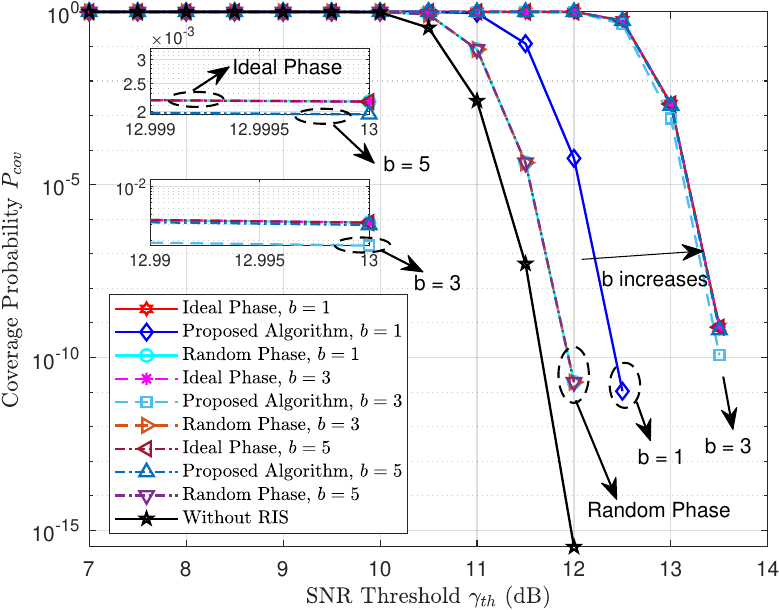}}
  \caption{ \label{fig:4}Coverage probability vs. SNR threshold. System parameters: $k=8,000$ m, $T=16,000$, $P = 30$ dBm, and $N = 50$.}
\end{figure}

\begin{figure}[!t] 
  \centering
  {\includegraphics[scale=0.6]{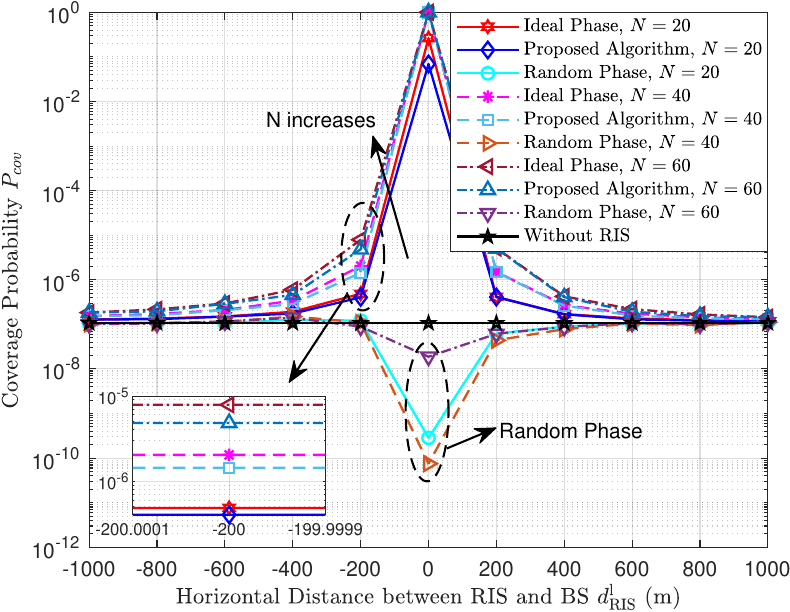}}
  \caption{ \label{fig:5}Coverage probability vs. horizontal distance between RIS and BS. System parameters: $k=11,000$ m, $T=22,000$, $P = 30$ dBm, $\gamma_{th} =10$ dB, and $b=2$.}
\end{figure}

Fi{}g.~{\ref{fig:5}} illustrates coverage probability against horizontal distance $d_{\rm{RIS}}^{\rm{l}}$  between the RIS and BS under different number of RIS elements. As expected, it is observed that coverage probability of the case ideal phase and the proposed algorithm firstly increases and then decreases. It is found that the system performance of the proposed algorithm is highest when $d_{\rm{RIS}}^{\rm{l}}=0$, namely, the center of the long axis of BS elliptical coverage area. This is owing to the fact that when RIS moves towards BS, path-loss of the reflection link changes accordingly, which leads to enhanced reflected signal and the benefits of the RIS are fully utilized. This shows that the system performance is sensitive to placement of the RIS. In addition, we observe that  narrow performance gap exists between the proposed algorithm and ideal phase case when $N=40$ and $N=60$, which means that coverage probability of the proposed algorithm can approach ideal phase case with increasing number of RIS elements.

Fi{}g.~{\ref{fig:6}} illustrates coverage probability against speed $v$ of HST. It can be observed that coverage probabilities under the four schemes increase  with  speed of HST. The reason is that HST can travel longer distance with the higher speed under the same circumstances, and lead it closer to BS in the current scenario, thus channel gain is improved, which results that received power at MR is enhanced and coverage probability is thus improved. When $v$ is further increased, coverage probabilities converge to the maximum value. We can also see that the proposed algorithm outperforms the cases without RIS and random phase except for the case of ideal phase. In addition, the minimum speed satisfying coverage probability constraint \eqref{yya} for the cases of ideal phase and the proposed algorithm are all $v = 300$ km/h, while it is $v = 400$ km/h and $v = 450$ km/h for the cases with random phase and without RIS, respectively. In other words, the proposed algorithm can reduce approximately 33.3\% of speed to establish fairly good communication between BS and MR, compared with the case without RIS, which further shows that it saves energy of HST system. This is because when HST gradually moves to BS, the channel of RIS-assisted HST communication becomes better, and channel gain of $h(t)$ increases. Under the same conditions, the driving distance required for HST to meet coverage probability constraint \eqref{yya} is shorter, thus the required speed is lower, and conversely, a faster speed is required in the case of without RIS deployment.
\begin{figure}[!t] 
	\centering
	{\includegraphics[scale=0.6]{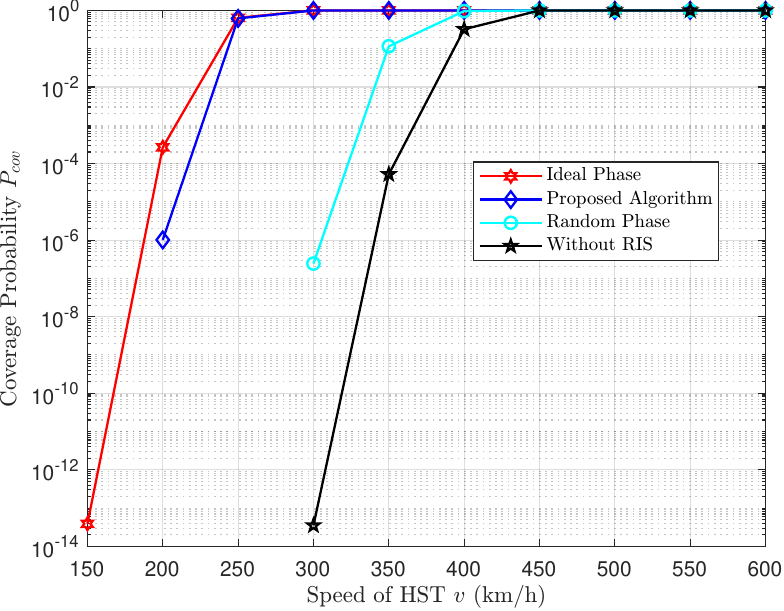}}
	\caption{ \label{fig:6}Coverage probability vs. speed of HST. System parameters: $k=7,000$ m, $T=14,000$, $P = 30$ dBm, $\gamma_{th} =10$ dB, $N$ = 50, and $b=2$.}
\end{figure}

\begin{figure}[!t] 
  \centering
  {\includegraphics[scale=0.6]{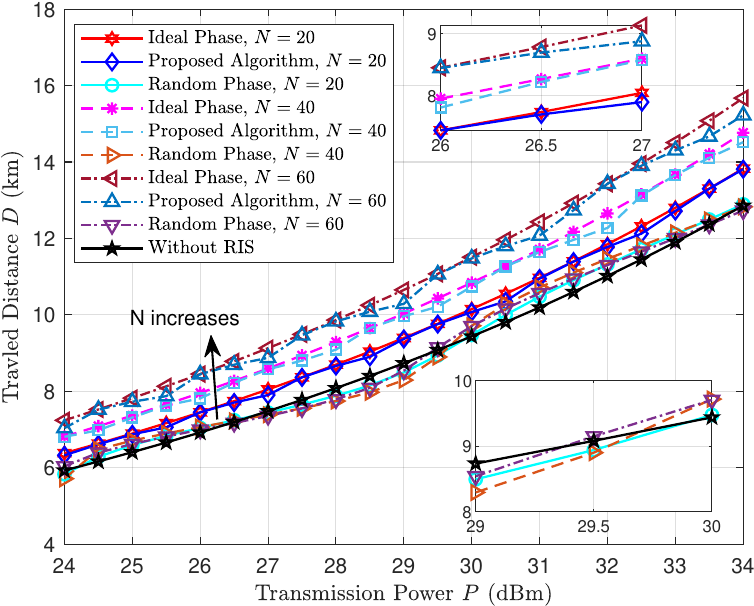}}
  \caption{ \label{fig:7}Travel distance vs. transmission power. System parameters: $k = 11,000$ m, $T=22,000$, $\gamma_{th} = 10$, dB and $b=2$.}
\end{figure}

\vspace{-8mm}
\subsection{Travel Distance}
In this subsection,  impact of the four schemes on travel distance is analyzed, which is calculated by \eqref{eq:P1} with constraint $P_{\rm{cov}}\left( t \right) \geq P_{\rm{th}}$. Fi{}g.~{\ref{fig:7}} shows travel distance against transmission power $P$ under different number of RIS elements. It can be observed that travel distance under the four schemes increases with transmission power, and the proposed algorithm outperforms the cases without RIS and random phase. The travel distance of the random phase case fluctuates with increasing number of RIS elements. The reason is that the distance of the reflection link that generates useful signal is random, however, the useful reflection signal is not necessarily strong unless RIS element phases are optimized, thus the travel distance is not necessarily increased, which can also provide an important insight that there is no advantage of deploying RIS unless RIS elements phases are optimized. We further observe that travel distance increases with number of RIS elements. This is because more RIS elements in system result that more signal paths and energy can be reflected to enhance signal quality at MR. For given $P=30$ dBm, when $N=20$, $N=40$ and $N=60$, we can observe that travel distances of the proposed algorithm more number of RIS are $6.93\%$, $13.72\%$ and $20.14\%$ higher that the case without RIS, respectively. This shows that introducing RIS can well enhance HST communication coverage. In addition, travel distance of the proposed algorithm is not strictly linearly increasing with transmission power, therefore, the degree of performance improvement is not the same under different transmission power.

\begin{figure}[!t] 
	\centering
	{\includegraphics[scale=0.6]{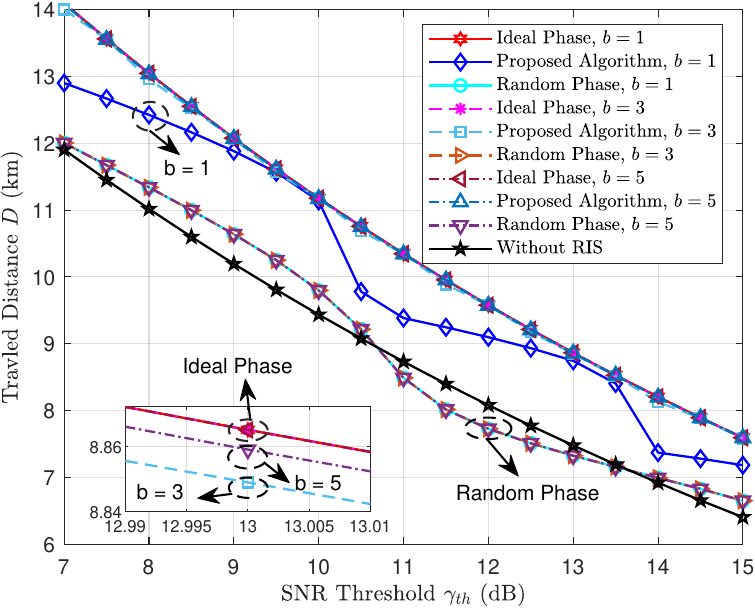}}
	\caption{ \label{fig:8}Travel distance vs. SNR threshold. System parameters: $k=8,000$ m, $T=16,000$, $P = 30$ dBm, and $N = 50$.}
\end{figure}
Fi{}g.~{\ref{fig:8}} shows travel distance against SNR threshold $\gamma_{th}$ under different number of RIS quantization bits. It is observed that travel distance of the four schemes decreases with  SNR threshold. This is because improving SNR threshold requirements leads to more communication outage events and lower coverage probability, thus the travel distance decreases accordingly. In addition, it is observed that the travel distance of the proposed algorithm increases with number of RIS quantization bits, and it is effected by $\gamma_{th}$. It is further observed that travel distance of the proposed algorithm has some fluctuations when $b=1$. This is because RIS can concentrate  incident signals with  more precise directions by using higher phase resolution $b$, and thus the case with higher $b$ can achieve more performance gain. This implies that performance of the proposed algorithm can better approach the  ideal phase case by utilizing higher phase resolution.

In Fi{}g.~{\ref{fig:9}}, we investigate  impact of horizontal distance $d_{\rm{RIS}}^{\rm{l}}$  between RIS and BS on travel distance under different number of RIS elements.  We plot travel distance curves with $d_{\rm{RIS}}^{\rm{l}}$ from ranging  $-1,000$ m to $1,000$ m for the four schemes. As expected, it is observed that travel distance of ideal phase and the proposed algorithm firstly increases and then decreases, and has the highest travel distance when $d_{\rm{RIS}}^{\rm{l}}=0$. This is because when RIS moves towards BS, path-loss of the reflection link changes accordingly, which leads to enhanced reflected signal at MR. When $d_{\rm{RIS}}^{\rm{l}}=200$ m, it is observed in Fi{}g.~{\ref{fig:14}} that the proposed algorithm can improved travel distance compared with the case without RIS by $0.45\%$, $0.89\%$ and $1.32\%$ for $N = 20$, $N = 40$ and $N = 60$, respectively, while which are  $7.31\%$, $11.61\%$ and $21.98\%$ higher than the case without RIS when $d_{\rm{RIS}}^{\rm{l}}=0$ m, respectively. It reveals the fact that the performance gain improvement is small when RIS is deployed far away from BS, and a larger number of RIS elements will be required for performance improvement. When RIS is close to the center of BS coverage area, the system performance can be significantly improved, and a small number of RIS components can achieve fairly good performance. This can not only improve coverage performance but also reduce hardware manufacturing cost.
\begin{figure}[!t] 
  \centering
  {\includegraphics[scale=0.6]{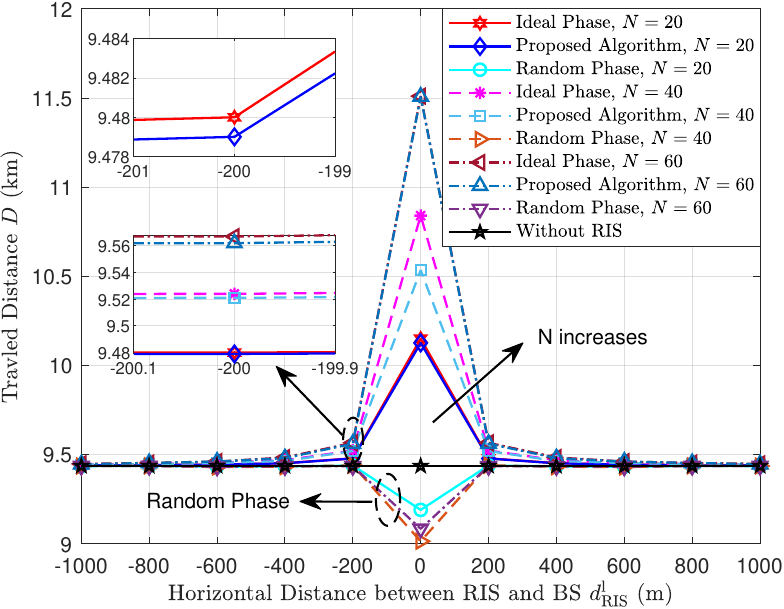}}
  \caption{ \label{fig:9}Travel distance vs. horizontal distance between RIS and BS. System parameters: $k=11,000$ m, $T=22,000$, $P = 30$ dBm, $\gamma_{th} =10$ dB, and $b=2$.}
\end{figure}

\begin{figure}[!t] 
  \centering
  {\includegraphics[scale=0.6]{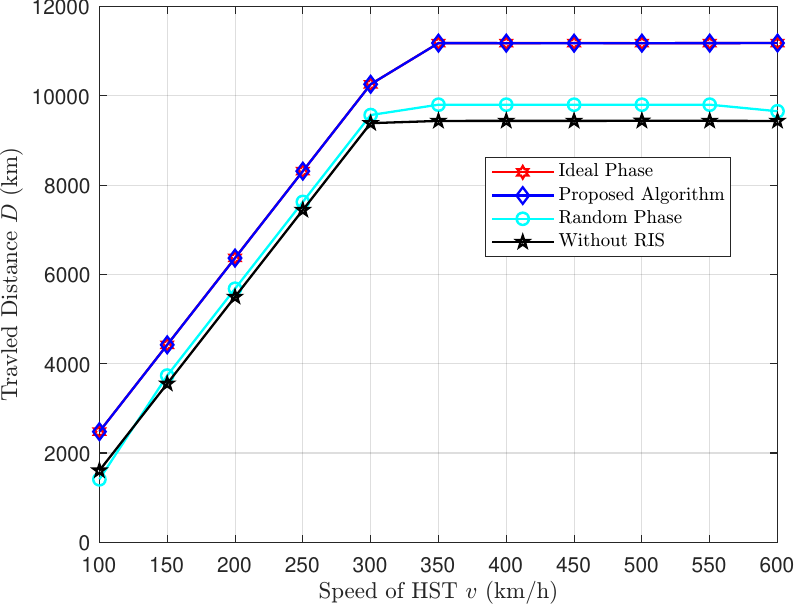}}
  \caption{ \label{fig:10}Travel distance vs. speed of HST. System parameters: $k=7,000$ m, $T=14,000$, $P = 30$ dBm, $\gamma_{th} =10$ dB, $N$ = 50, and $b=2$.}
\end{figure}
Fi{}g.~{\ref{fig:10}} illustrates travel distance against speed $v$ of HST. It is observed that travel distance increases with speed of HST. When $v$ is further increased, travel distance converges to the maximum value. This is because within the same time, it takes more time for HST to reach the position which can satisfy coverage probability constraint with lower speed, resulting in a shorter travel distance under coverage probability constraint. When the speed becomes higher, the time it takes is shorter and the travel distance will become longer. When the speed increases to a certain value, the time it takes for HST to run the distance that satisfies coverage probability constraint will not exceed set time range, that is, travel distance of the current conditions reaches the upper limit, thus, it will not  further increase with speed. In addition, when speed does not exceed $300$ km/h, it is observed that the proposed algorithm can improved travel distance compared with the case without RIS by 9.28\%, while it is 18.44\% higher than the case without RIS when speed exceeds $350$ km/h. It reveals that the performance gain is more obvious at a higher speed.

\vspace{-3mm}
\subsection{Transmission Rate}
In this subsection, we investigate impact of the four schemes on transmission rate. The transmission rate in time slot $t$ can be written as
\begin{equation}
  \begin{array}{l} \label{eq:rate}
    C \left( t \right) \\
    =B\log_{10}\left(1 + \frac{P\left|h_{\rm{BM}} \left( t \right) + \sum_{n=1}^N{h^{n}_{{\rm{RM}} }\left( t \right) e^{j\theta _n\left( t \right)}h^{n}_{{\rm{BR}} }}\left( t \right)  \right|^2 }{\sigma^2}  \right) \\
    = B\log_{10}\left( 1 + \gamma\left( t \right) \right).
  \end{array}
\end{equation}

\begin{figure}[!t] 
  \centering
  {\includegraphics[scale=0.6]{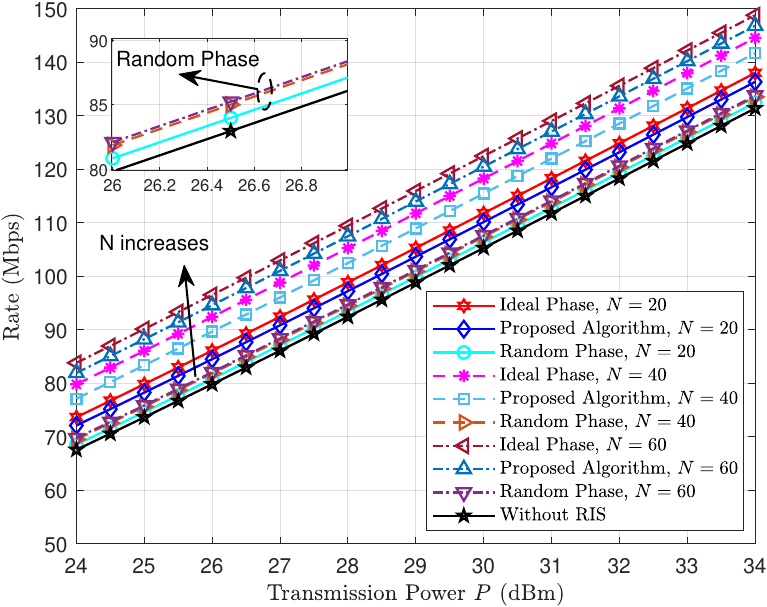}}
  \caption{ \label{fig:11}Rate vs. BS transmission power. System parameters: $k = 11,000$ m, $T=22,000$, $\gamma_{th} = 10$, dB and $b=2$.}
\end{figure}

Fi{}g.~{\ref{fig:11}} illustrates transmission rate against transmission power $P$ under different numbers of RIS elements. It is observed that transmission rate increases with transmission power. This is because SNR increases with transmission power and thus results that transmission rate increases. Moreover, it is observed that transmission rate increases with number of RIS elements. Compared with the case without RIS, transmission rate of the proposed algorithm achieves a gain of $3.75\%$, $7.84\%$ and $11.73\%$ when  $N = 20$, $N = 40$ and $N = 60$, respectively for $P=34$ dBm, whereas it can achieves up to $6.71\%$, $14.13\%$ and $21.15\%$ for $P=24$ dBm, respectively. This is because increasing number of RIS elements has higher impact on the improvement of channel gain compared with increasing transmission power. It also reveals that using RIS can achieve improved system performance and reduce power consumption.

Fi{}g.~{\ref{fig:12}} shows transmission rate against SNR threshold $\gamma_{th}$ under different numbers of RIS quantization bits. As excepted, it is observed that transmission rate is not affected by increasing transmission power. We further observe that transmission rate of the proposed algorithm increases with number of RIS quantization bits $b$. The proposed algorithm can generally achieve $7.61\%$, $12.12\%$ and $13.17\%$ transmission rate improvement, respectively, compared with the case without RIS when $b = 1$, $b = 3$ and $b = 5$. This means that transmission rate is significantly improved with higher phase resolution $b$. Increasing $b$ from $1$ to $3$ and to $5$, the proposed algorithm achieves a gain of $4.75\%$ and $5.14\%$ respectively, which indicates that system performance generally increases slowly with $b$. This reveals that the system performance tends to be saturated when number of quantization bits exceeds a certain value.
\begin{figure}[!t] 
  \centering
  {\includegraphics[scale=0.6]{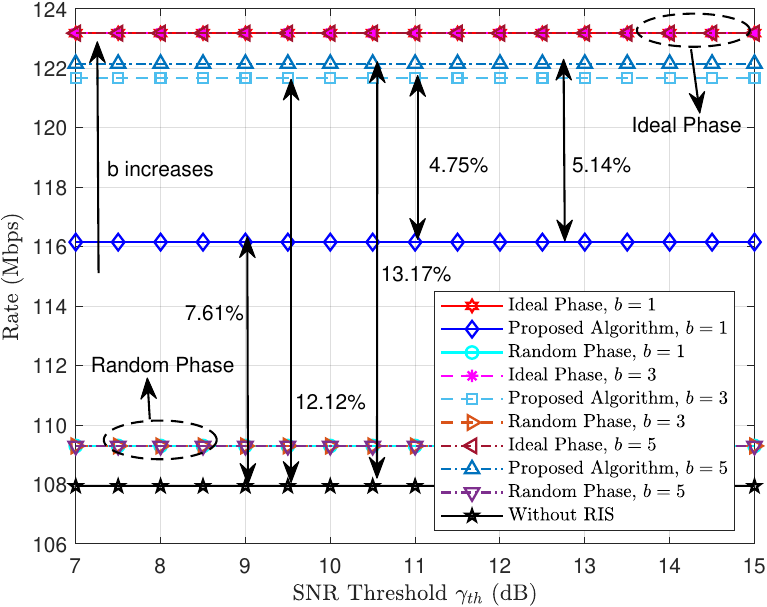}}
  \caption{ \label{fig:12}Rate vs. SNR threshold. System parameters: $k=8,000$ m, $T=16,000$, $P = 30$ dBm, and $N = 50$.}
\end{figure}

Fi{}g.~{\ref{fig:13}} illustrates transmission rate against horizontal distance $d_{\rm{RIS}}^{\rm{l}}$ between RIS and BS. It is observed that transmission rate of using ideal phase and the proposed algorithm firstly increases and then decreases, and have the highest travel distance when $d_{\rm{RIS}}^{\rm{l}}=0$. Transmission rate for the cases ideal phase and the proposed algorithm increases with  number of RIS elements. The reason is that the channel gain is improved with increasing number of RIS elements and RIS is close to the center of BS coverage area. Under the default setting of $d_{\rm{RIS}}^{\rm{l}}=0$, the proposed algorithm can achieve $7.31\%$, $11.67\%$ and $21.98\%$ transmission rate improvement, respectively, compared with the case without RIS when $N = 20$, $N = 40$, and $N = 60$.
\begin{figure}[!t] 
  \centering
  {\includegraphics[scale=0.6]{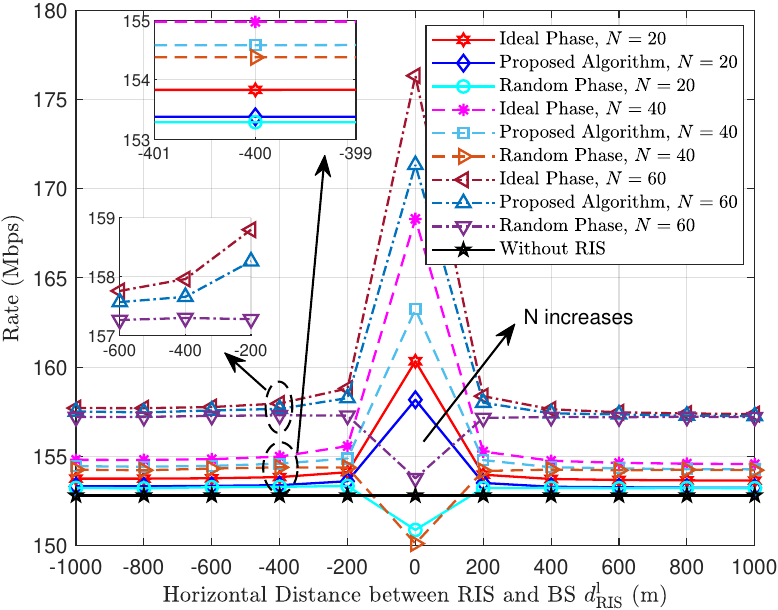}}
  \caption{ \label{fig:13}Rate vs. horizontal distance between RIS and BS. System parameters: $k=11,000$ m, $T=22,000$, $P = 30$ dBm, $\gamma_{th} =10$ dB, and $b=2$.}
\end{figure}

\section{Conclusions}
In this paper, we investigate coverage performance and travel distance maximization of downlink SISO RIS-assisted HST communication system. An closed-form expression of coverage probability is derived according to CSI, and travel distance maximization problem is solved by jointly design of RIS phase and RIS placement. We have analyzed the impacts of some key system parameters including transmission power, SNR threshold, number of RIS elements, number of RIS quantization bits, horizontal distance between BS and RIS, and speed of HST on coverage probability, travel distance, and transmission rate of HST communications. Numerical results have demonstrated that improved coverage performance and larger travel distance can be achieved by deployment of RIS and  efficiency of the proposed algorithm in verified in terms of coverage probability, travel distance, and transmission rate. The results in this paper can serve as a guidance for RIS-assisted HST communication coverage analysis and system design.

\section*{Appendix A \\ The proof of Theorem 1}  
Substituting \eqref{eq:hbm}, \eqref{eq:hbr} and \eqref{eq:hrm} into \eqref{eq:h}, we have \eqref{eq:h2},
\begin{figure*}[!h]
  \begin{align}\label{eq:h2}
      h\left( t \right) &= \rho _{\mathrm{BM}}\bar{h}_{ \rm{BM}}\left( t \right) + \varrho _{{\mathrm{BM}}}\tilde{h}_{ \rm{BM}}\left( t \right)  + \sum_{n=1}^N{\rho _{\mathrm{RM}}\rho _{\mathrm{BR}}\bar{h}^{n}_{{\rm{RM}} }\left( t \right) e^{j\theta _n\left( t \right)}\bar{h}^{n}_{{\rm{BR}} }} + \sum_{n=1}^N{\rho _{\mathrm{RM}}\varrho _{{\mathrm{BR}}}\bar{h}^{n}_{{\rm{RM}} }\left( t \right) e^{j\theta _n\left( t \right)}\tilde{h}^{n}_{{\rm{BR}} }}\left( t \right) \\ 
      &+ \sum_{n=1}^N{\varrho _{\mathrm{RM}}\rho _{{\mathrm{BR}}}\tilde{h}^{n}_{{\rm{RM}} }\left( t \right) e^{j\theta _n\left( t \right)}\bar{h}^{n}_{{\rm{BR}} }} + \sum_{n=1}^N{\varrho _{\mathrm{RM}}\varrho _{{\mathrm{BR}}}\tilde{h}^{n}_{{\rm{RM}} }\left( t \right) e^{j\theta _n\left( t \right)}\tilde{h}^{n}_{{\rm{BR}} }}\left( t \right) \nonumber \\
      &= h_1\left( t \right) + h_2\left( t \right)+ h_3\left( t \right)+ h_4\left( t \right)+ h_5\left( t \right) +h_6\left( t \right),  \nonumber 
  \end{align}
  \hrulefill 
\end{figure*}
where $\rho_{\mathrm{BM}}=\sqrt{\frac{\kappa _{{\mathrm{BM}}}}{\kappa _{{\mathrm{BM}}}+1}}$, $\varrho _{{\mathrm{BM}}}=\sqrt{\frac{1}{\kappa _{\mathrm{BM}}+1}}$, $\rho _{\mathrm{BR}}=\sqrt{\frac{\kappa _{{\mathrm{BR}}}}{\kappa _{{\mathrm{BR}}}+1}}$, $\varrho _{{\mathrm{BR}}}=\sqrt{\frac{1}{\kappa _{{\mathrm{BR}}}+1}}$, $\rho _{\mathrm{RM}}=\sqrt{\frac{\kappa _{\mathrm{RM}}}{\kappa _{\mathrm{RM}}+1}}$, $\varrho _{{\mathrm{RM}}}=\sqrt{\frac{1}{\kappa _{\mathrm{RM}}+1}}$, $h_1\left( t \right) = \rho _{\mathrm{BM}}\bar{h}_{ \rm{BM}}\left( t \right)$, $h_2\left( t \right) = \varrho _{{\mathrm{BM}}}\tilde{h}_{ \rm{BM}}\left( t \right)$, $h_3\left( t \right) = \sum_{n=1}^N{\rho _{\mathrm{RM}}\rho _{\mathrm{BR}}\bar{h}^{n}_{{\rm{RM}} }\left( t \right) e^{j\theta _n\left( t \right)}\bar{h}^{n}_{{\rm{BR}}}}$, $h_4\left( t \right) = \sum_{n=1}^N{\rho _{\mathrm{RM}}\varrho _{{\mathrm{BR}}}\bar{h}^{n}_{{\rm{RM}} }\left( t \right) e^{j\theta _n\left( t \right)}\tilde{h}^{n}_{{\rm{BR}} }}\left( t \right)$, $h_5\left( t \right) = \sum_{n=1}^N{\varrho _{\mathrm{RM}}\rho _{{\mathrm{BR}}}\tilde{h}^{n}_{{\rm{RM}} }\left( t \right) e^{j\theta _n\left( t \right)}\bar{h}^{n}_{{\rm{BR}}}}$ and $h_6\left( t \right) = \sum_{n=1}^N{\varrho _{\mathrm{RM}}\varrho _{{\mathrm{BR}}}\tilde{h}^{n}_{{\rm{RM}} }\left( t \right) e^{j\theta _n\left( t \right)}\tilde{h}^{n}_{{\rm{BR}} }}\left( t \right)$.

Note that, the LoS components of  BS-MR link,  RIS-MR link and  BS-RIS link depend on the corresponding link distances. For a given location, components  $h_1\left( t \right)$ and  $h_3\left( t \right)$ of \eqref{eq:h2} turn to be deterministic. Since the NLoS component $\tilde{h}_{ \rm{BM}}\left( t \right)$ follows a complex Gaussian distribution with zero mean and variance $PL _{\rm{NLoS}}^{BM}\left( t \right)$, $\tilde{h}_{ \rm{RM}}\left( t \right)$ and $\tilde{h}_{ \rm{BR}}\left( t \right)$ follow a complex Gaussian distribution with zero mean and variance $PL _{\rm{NLoS,RM}}^{n}\left( t \right)$, $PL _{\rm{NLoS,BR}}^{n}\left( t \right)$, respectively, the parts $h_2\left( t \right)$, $h_4\left( t \right)$, $h_5\left( t \right)$ and $h_6\left( t \right)$ of \eqref{eq:h2} also follow a Gaussian distribution.

The expectation of $h\left( t \right)$ can be written as in \eqref{eq:eh2},
\begin{figure*}[!h]
  \begin{align} \label{eq:eh2}
    &\mu_h \left(t\right) \ \triangleq \mathbb{E}\left\{ h\left( t \right)\right\} = \mathbb{E}\left\{ \rho _{\mathrm{BM}}\bar{h}_{ \rm{BM}}\left( t \right) \right\} + \mathbb{E}\left\{ \varrho _{{\mathrm{BM}}}\tilde{h}_{ \rm{BM}}\left( t \right) \right\}+ \mathbb{E}\left\{ \sum_{n=1}^N{\rho _{\mathrm{RM}}\rho _{\mathrm{BR}}\bar{h}^{n}_{{\rm{RM}} }\left( t \right) e^{j\theta _n\left( t \right)}\bar{h}^{n}_{{\rm{BR}} }} \right\}  \\ \nonumber
  &+ \mathbb{E}\left\{ \sum_{n=1}^N{\rho _{\mathrm{RM}}\varrho _{{\mathrm{BR}}}\bar{h}^{n}_{{\rm{RM}} }\left( t \right) e^{j\theta _n\left( t \right)}\tilde{h}^{n}_{{\rm{BR}} }}\left( t \right) \right\} + \mathbb{E}\left\{  \sum_{n=1}^N{\varrho _{\mathrm{RM}}\rho _{{\mathrm{BR}}}\tilde{h}^{n}_{{\rm{RM}} }\left( t \right) e^{j\theta _n\left( t \right)}\bar{h}^{n}_{{\rm{BR}} }} \right\} \\ \nonumber
  &+ \mathbb{E}\left\{ \sum_{n=1}^N{\varrho _{\mathrm{RM}}\varrho _{{\mathrm{BR}}}\tilde{h}^{n}_{{\rm{RM}} }\left( t \right) e^{j\theta _n\left( t \right)}\tilde{h}^{n}_{{\rm{BR}} }}\left( t \right) \right\} \\ \nonumber
  &=\mathbb{E}\left\{ \rho _{\mathrm{BM}}\bar{h}_{ \rm{BM}}\left( t \right) \right\} + \mathbb{E}\left\{ \sum_{n=1}^N{\rho _{\mathrm{RM}}\rho _{\mathrm{BR}}\bar{h}^{n}_{{\rm{RM}} }\left( t \right) e^{j\theta _n\left( t \right)}\bar{h}^{n}_{{\rm{BR}} }} \right\}  = \rho _{\mathrm{BM}}\bar{h}_{ \rm{BM}}\left( t \right) +\sum_{n=1}^N{\rho _{\mathrm{RM}}\rho _{\mathrm{BR}}\bar{h}^{n}_{{\rm{RM}} }\left( t \right) e^{j\theta _n\left( t \right)}\bar{h}^{n}_{{\rm{BR}} }} \\  \nonumber
  &=\rho _{\rm{BM} }\sqrt{PL _{\rm{BM} }\left(t\right)}e^{-j\theta ^{ \rm{BM}}\left(t\right)} +\sum_{n=1}^N\rho_{\rm{RM} }\rho _{\rm{BR} }\sqrt{PL _{ \rm{RM}}^{n}\left(t\right)}\sqrt{PL _{ \rm{BR}}^{n}}e^{j\left( \theta _n\left(t\right)-\theta_{\rm{RM}}^{n}\left(t\right)-\theta_{\rm{BR}}^{n} \right)}, \nonumber
  \end{align}
  \hrulefill 
\end{figure*}

and the variance of $h\left( t \right)$  is derived as in \eqref{eq:hsigama},
\begin{figure*}[!ht]
  \begin{align} \label{eq:hsigama}
    \sigma _{h}^{2}\left(t\right) & \triangleq \rm{var}\left\{ h\left( t \right)\right\} = \rm{var}\left\{ \rho _{\mathrm{BM}}\bar{h}_{ \rm{BM}}\left( t \right) \right\} + \rm{var}\left\{ \varrho _{{\mathrm{BM}}}\tilde{h}_{ \rm{BM}}\left( t \right) \right\}  + \rm{var}\left\{ \sum_{n=1}^N{\rho _{\mathrm{RM}}\rho _{\mathrm{BR}}\bar{h}^{n}_{{\rm{RM}} }\left( t \right) e^{j\theta _n\left( t \right)}\bar{h}^{n}_{{\rm{BR}} }} \right\}  \\ \nonumber
   & + \rm{var}\left\{ \sum_{n=1}^N{\rho _{\mathrm{RM}}\varrho _{{\mathrm{BR}}}\bar{h}^{n}_{{\rm{RM}} }\left( t \right) e^{j\theta _n\left( t \right)}\tilde{h}^{n}_{{\rm{BR}} }}\left( t \right) \right\} + \rm{var}\left\{  \sum_{n=1}^N{\varrho _{\mathrm{RM}}\rho _{{\mathrm{BR}}}\tilde{h}^{n}_{{\rm{RM}} }\left( t \right) e^{j\theta _n\left( t \right)}\bar{h}^{n}_{{\rm{BR}} }} \right\} \\ \nonumber
   & + \rm{var}\left\{ \sum_{n=1}^N{\varrho _{\mathrm{RM}}\varrho _{{\mathrm{BR}}}\tilde{h}^{n}_{{\rm{RM}} }\left( t \right) e^{j\theta _n\left( t \right)}\tilde{h}^{n}_{{\rm{BR}} }}\left( t \right) \right\} \\  \nonumber
  &=\rm{var}\left\{ \varrho _{{\mathrm{BM}}}\tilde{h}_{ \rm{BM}}\left( t \right) \right\} + \rm{var}\left\{ \sum_{n=1}^N{\varrho _{\mathrm{RM}}\varrho _{{\mathrm{BR}}}\tilde{h}^{n}_{{\rm{RM}} }\left( t \right) e^{j\theta _n\left( t \right)}\tilde{h}^{n}_{{\rm{BR}} }}\left( t \right) \right\} \\  \nonumber
  &= \varrho _{{\mathrm{BM}}}^2\rm{var}\left\{ \tilde{h}_{ \rm{BM}}\left( t \right) \right\} +\varrho _{\mathrm{RM}}^2\varrho _{{\mathrm{BR}}}^2\rm{var}\left\{ \sum_{n=1}^N{\tilde{h}^{n}_{{\rm{RM}} }\left( t \right) e^{j\theta _n\left( t \right)}\tilde{h}^{n}_{{\rm{BR}} }}\left( t \right) \right\} \\  \nonumber
  &=\varrho_{\rm{BM}}^2PL _{\rm{NLoS}}^{\rm{BM}}\left(t\right) +\sum_{n=1}^{N}\varrho_{\rm{RM}}^2\varrho_{\rm{BR}}^2PL _{\rm{NLoS,RM}}^{n}\left(t\right)PL _{{\rm{NLoS,BR}}}^{n}\left(t\right), \nonumber
  \end{align}
    \hrulefill
  \end{figure*}
thus, $h\left( t \right)$ is proved to follow a complex-valued Gaussian distribution, $h\left(t\right) \sim \mathcal{C} \mathcal{N} \left( \mu_h\left(t\right),\sigma_{h}^2\left(t\right) \right)$. This completes the proof.

\vspace{-3.5mm}
\section*{Appendix B \\ The proof of Theorem 2}
Due to the Gaussian channel proved above, $h\left(t\right) \sim \mathcal{C} \mathcal{N} \left( \mu_h\left(t\right),\sigma_{h}^2\left(t\right) \right) $. Therefore, $\frac{  \left| h\left(t\right) \right|^2 }{\sigma_{h}^2\left(t\right)}$ follows the non-central chi-squared distribution, i.e.,  $\chi ^2(\nu,\zeta \left( t \right)) $, with the degree of freedom $\nu = 2$, and the  non-centrality parameter is as in \eqref{eq:notcenter}.
\begin{figure*}[!ht]
  \begin{equation} \label{eq:notcenter}
    \zeta \left( t \right) = \frac{ \left| \mu_h\left(t\right)\right|^2 }{\sigma_{h}^2\left(t\right)} = \frac{\left|\rho _{\rm{BM} }\sqrt{PL _{\rm{BM} }\left(t\right)}e^{-j\theta _{ \rm{BM}}\left(t\right)} +\sum_{n=1}^N\rho_{\rm{RM} }\rho _{\rm{BR} }\sqrt{PL _{ \rm{RM}}^{n}\left(t\right)}\sqrt{PL _{ \rm{BR}}^{n}} e^{j\left( \theta _n\left( {t} \right) -\theta _{\mathrm{RM}}^{n}\left( {t} \right) -\theta _{\mathrm{BR}}^{n} \right)} \right|^2}{\varrho_{\rm{BM}}^2PL _{\rm{NLoS}}^{\rm{BM}}\left(t\right) +\sum_{n=1}^{N}\varrho_{\rm{RM}}^2\varrho_{\rm{BR}}^2PL _{\rm{NLoS,RM}}^{n}\left(t\right)PL _{{\rm{NLoS,BR}}}^{n}\left(t\right) }.
  \end{equation}
  \hrulefill
\end{figure*}

\vspace{-3mm}
With the corresponding CDF of $\chi ^2(\nu,\zeta \left( t \right)) $, $P_{\rm{out}}\left( t \right)$ defined in \eqref{eq:pcov} is given by
\begin{align}
  P_{\rm{out}}\left( t \right) &= \mathrm{Pr}\left( \left| h\left( t \right) \right|^2<\frac{\gamma _{\rm{th}}}{\bar{\gamma}} \right)  = \mathrm{Pr}\left(  \frac{\left| h\left( t \right) \right|^2 }{\sigma_{h}^{2}\left( t \right)} <\frac{\gamma _{\rm{th}}}{\bar{\gamma}\sigma_{h}^{2}\left( t \right)} \right) \\ \nonumber
  &=1 - Q_{\frac{\nu}{2}}\left(\sqrt{\zeta\left( t \right)}, \sqrt{\gamma_0\left( t \right)}\right),
\end{align}
where $\gamma_0 = \frac{\gamma_{th}}{\bar{\gamma}\sigma _{h}^{2}\left( t \right)}$, and $Q_m \left(a,b \right)$ is the Marcum Q-function defined in \cite{t8}. Thus, coverage probability defined in \eqref{eq:pcov} can be rewrite as 
\begin{align}
  P_{\rm{cov}} &= 1 - P_{\rm{out}} = Q_{1}\left(\sqrt{\zeta\left( t \right)}, \sqrt{\gamma_0\left( t \right)}\right).
\end{align}
This completes the proof.

\begin{IEEEbiography}[{\includegraphics[width=1in,height=1.25in,clip,keepaspectratio]{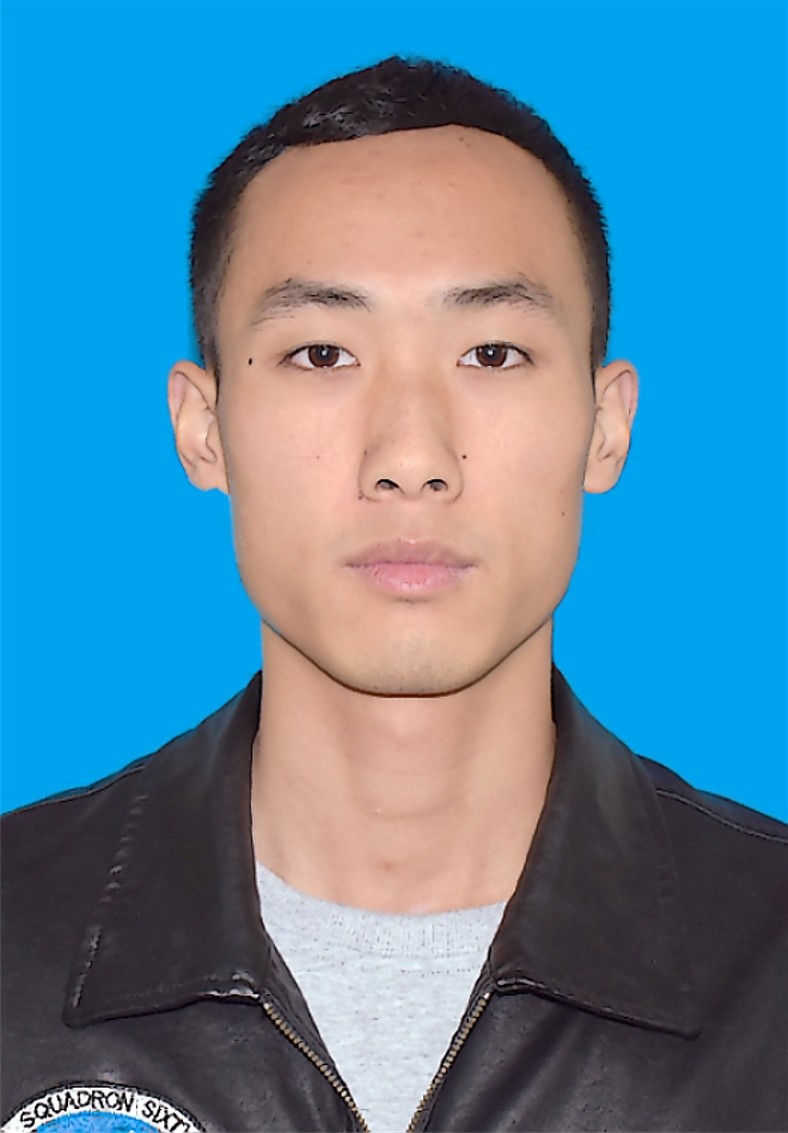}}]{Changzhu Liu}(Graduate Student Member, IEEE) received the B.Eng. degree in communication engineering from Huaihua University, Huaihua, Hunan, China, in 2017, the M.Sc. degree from the Chongqing University of Posts and Telecommunications (CQUPT), Chongqing, China, in 2020. He i  currently pursuing the Ph.D.degree with the State
Key Laboratory of Advanced Rail Autonomous Operation, Beijing Jiaotong University, Beijing, China. His current research interests include resource management of reconfigurable intelligent surface communications.
\vspace{-10mm}
\end{IEEEbiography}

\begin{IEEEbiography}[{\includegraphics[width=1in,height=1.25in,clip,keepaspectratio]{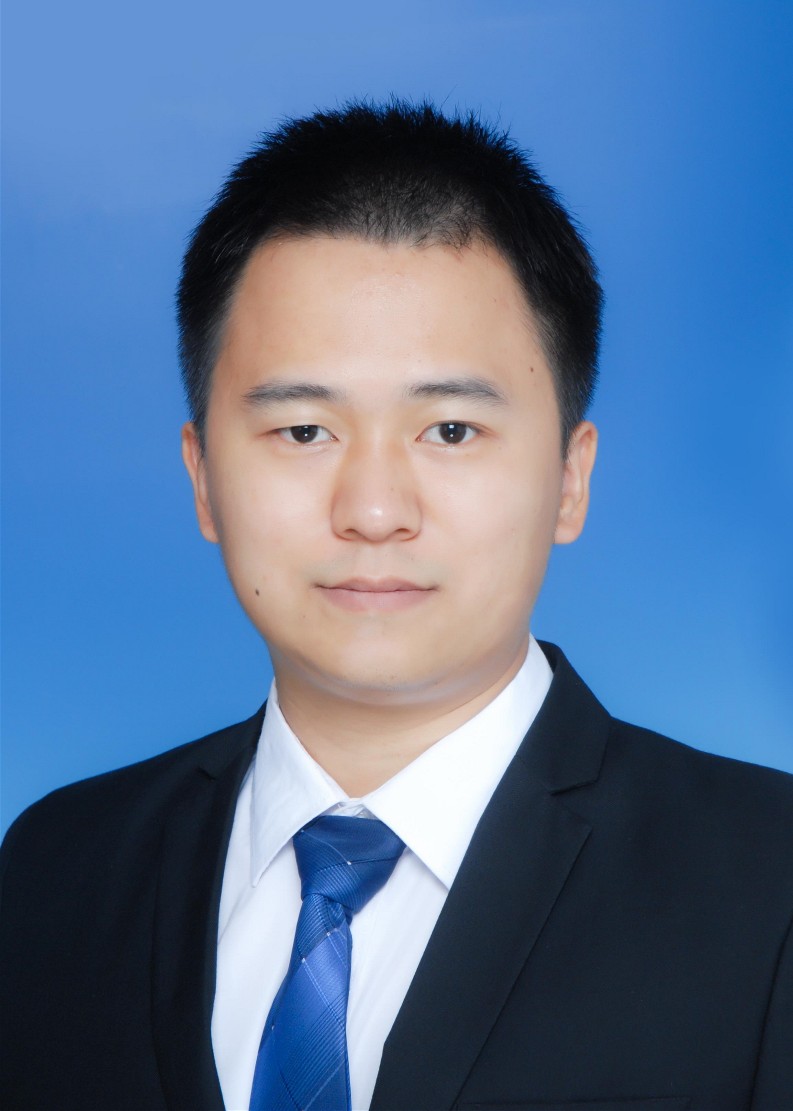}}]{Ruisi He}
(Senior Member, IEEE) received the B.E. and Ph.D. degrees from Beijing Jiaotong University (BJTU), Beijing, China, in 2009 and 2015, respectively.
  
Dr. He is currently a Professor with the State Key Laboratory of Advanced Rail Autonomous Operation and the School of Electronics and Information Engineering, BJTU. Dr. He has been a Visiting Scholar in Georgia Institute of Technology, USA, University of Southern California, USA, and Universit\'e Catholique de Louvain, Belgium. His research interests include wireless propagation channels, railway and vehicular communications, 5G and 6G communications. He has authored/co-authored 8 books, 4 book chapters, more than 200 journal and conference papers, as well as several patents.
  
Dr. He has been an Editor of the IEEE Transactions on Communications, the IEEE Transactions on Wireless Communications, the IEEE Transactions on Antennas and Propagation, the IEEE Antennas and Propagation Magazine, the IEEE Communications Letters, the IEEE Open Journal of Vehicular Technology, and a Lead Guest Editor of the IEEE Journal on Selected Area in Communications and the IEEE Transactions on Antennas and Propagation. He served as the Early Career Representative (ECR) of Commission C, International Union of Radio Science (URSI). He received the URSI Issac Koga Gold Medal in 2020, the IEEE ComSoc Asia-Pacific Outstanding Young Researcher Award in 2019, the URSI Young Scientist Award in 2015, and several Best Paper Awards in IEEE journals and conferences.
\end{IEEEbiography}

\begin{IEEEbiography}[{\includegraphics[width=1in,height=1.25in,clip,keepaspectratio]{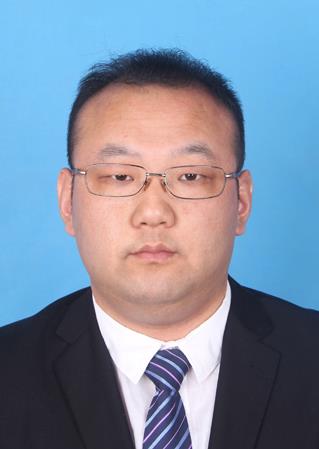}}]{Yong Niu}(Senior Member, IEEE)received the B.E. degree in electrical engineering from Beijing Jiaotong University, Beijing, China, in 2011, and the Ph.D. degree in electronic engineering from Tsinghua University, Beijing, China, in 2016.

From 2014 to 2015, he was a Visiting Scholar with the University of Florida, Gainesville, FL, USA. He is currently an Associate Professor with the State Key Laboratory of Rail Traffic Control and Safety, Beijing Jiaotong University. His research interests are in the areas of networking and communications, including millimeter wave communications, device-to-device communication, medium access control, and software-defined networks. He received the Ph.D. National Scholarship of China in 2015, the Outstanding Ph.D. Graduates and Outstanding Doctoral Thesis of Tsinghua University in 2016, the Outstanding Ph.D. Graduates of Beijing in 2016, and the Outstanding Doctorate Dissertation Award from the Chinese Institute of Electronics in 2017. He has served as Technical Program Committee Member for IWCMC 2017, VTC2018-Spring, IWCMC 2018, INFOCOM 2018, and ICC 2018. He was the Session Chair for IWCMC 2017. He was the recipient of the 2018 International Union of Radio Science Young Scientist Award.
\vspace{30mm}
\end{IEEEbiography}

\begin{IEEEbiography}[{\includegraphics[width=1in,height=1.25in,clip,keepaspectratio]{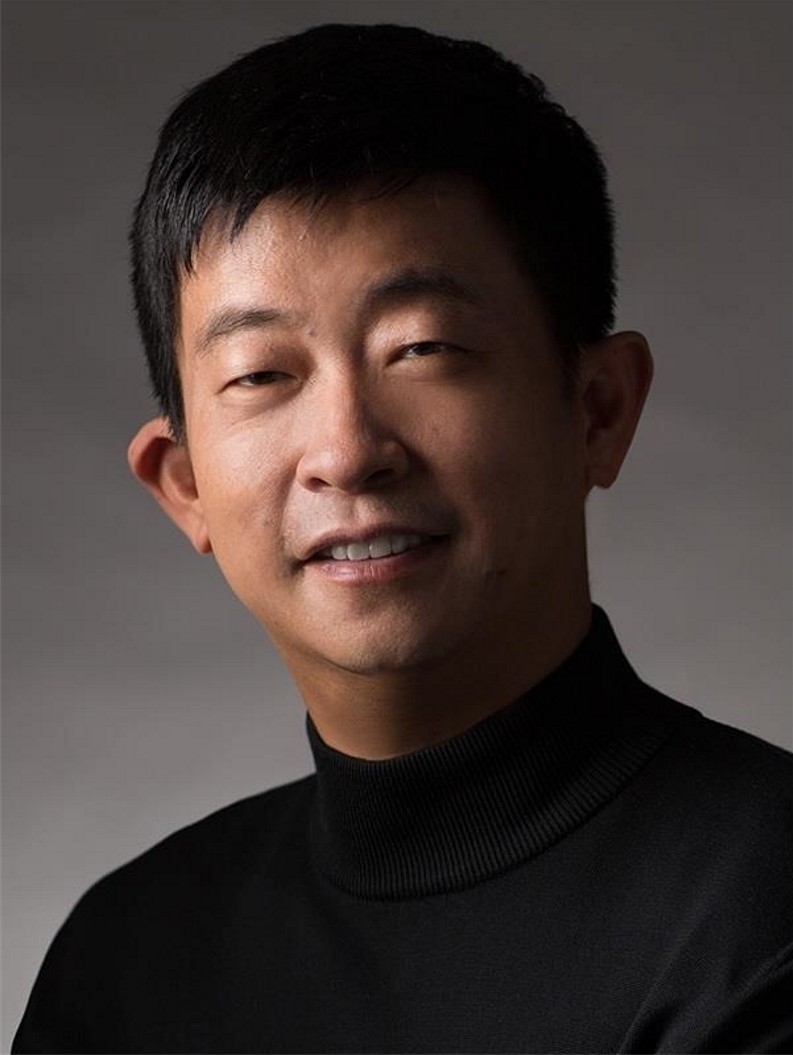}}]{Zhu Han} (Fellow, IEEE) received the B.S. degree in electronic engineering from Tsinghua University, in 1997, and the M.S. and Ph.D. degrees in electrical and computer engineering from the University of Maryland, College Park, in 1999 and 2003, respectively. 

From 2000 to 2002, he was an R\&D Engineer of JDSU, Germantown, Maryland. From 2003 to 2006, he was a Research Associate at the University of Maryland. From 2006 to 2008, he was an assistant professor at Boise State University, Idaho. Currently, he is a John and Rebecca Moores Professor in the Electrical and Computer Engineering Department as well as in the Computer Science Department at the University of Houston, Texas. Dr. Han's main research targets on the novel game-theory related concepts critical to enabling efficient and distributive use of wireless networks with limited resources. His other research interests include wireless resource allocation and management, wireless communications and networking, quantum computing, data science, smart grid, carbon neutralization, security and privacy.  Dr. Han received an NSF Career Award in 2010, the Fred W. Ellersick Prize of the IEEE Communication Society in 2011, the EURASIP Best Paper Award for the Journal on Advances in Signal Processing in 2015, IEEE Leonard G. Abraham Prize in the field of Communications Systems (best paper award in IEEE JSAC) in 2016, and several best paper awards in IEEE conferences. Dr. Han was an IEEE Communications Society Distinguished Lecturer from 2015-2018, AAAS fellow since 2019, and ACM distinguished Member since 2019. Dr. Han is a 1\% highly cited researcher since 2017 according to Web of Science. Dr. Han is also the winner of the 2021 IEEE Kiyo Tomiyasu Award (an IEEE Field Award), for outstanding early to mid-career contributions to technologies holding the promise of innovative applications, with the following citation: ``for contributions to game theory and distributed management of autonomous communication networks."
\end{IEEEbiography}

\begin{IEEEbiography}[{\includegraphics[width=1in,height=1.25in,clip,keepaspectratio]{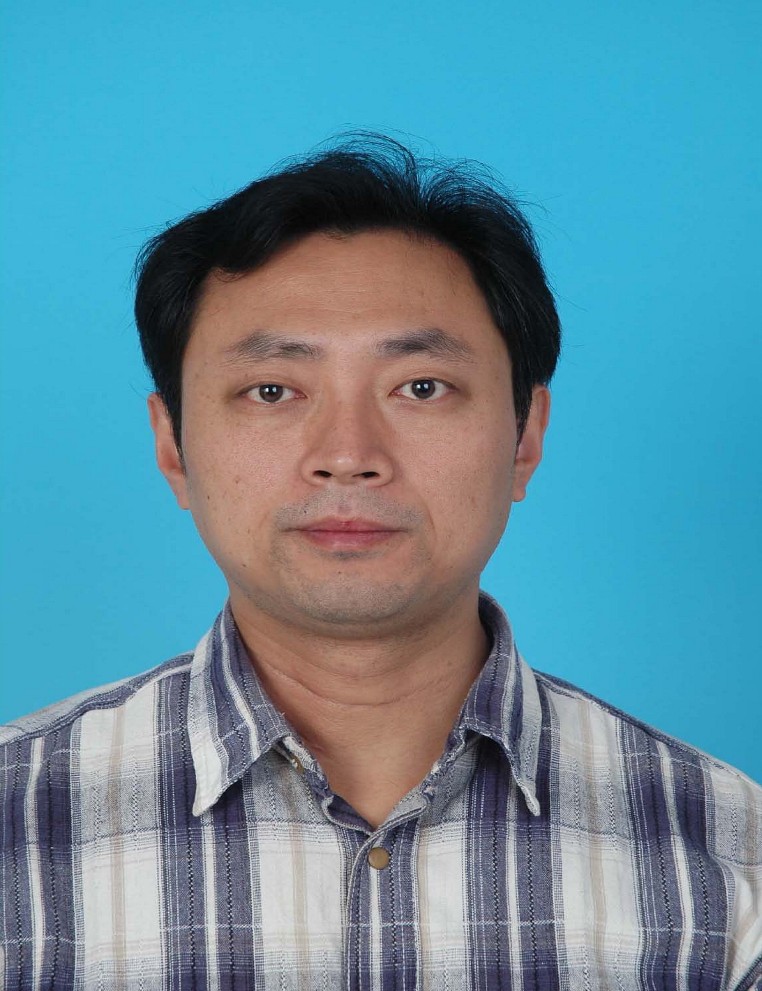}}]{Bo Ai}(Fellow, IEEE) received the master's and Ph.D. degrees from Xidian University, China. He received the honor of Excellent Postdoctoral Research Fellow from Tsinghua University in 2007. He was a Visiting Professor with the Electrical Engineering Department, Stanford University, Stanford, CA, USA, in 2015. He is currently a Full Professor with Beijing Jiaotong University, where he is also the Deputy Director of the State Key Laboratory of Rail Traffic Control and Safety and the Deputy Director of the International Joint Research Center. He is one of the directors for Beijing ``Urban Rail Operation Control System" International Science and Technology Cooperation Base, and the backbone member of the Innovative Engineering based jointly granted by the Chinese Ministry of Education and the State Administration of Foreign Experts Affairs.

He is the Research Team Leader of 26 national projects. He holds 26 invention patents. His interests include the research and applications of channel measurement and channel modeling and dedicated mobile communications for rail traffic systems. He has authored or co-authored eight books and authored over 300 academic research papers in his research area. Five papers have been the ESI highly cited paper. He has won some important scientific research prizes. He has been notified by the Council of Canadian Academies that based on the Scopus database, he has been listed as one of the Top 1\% authors in his field all over the world. He has also been feature interviewed by the \emph{IET Electronics Letters}.

Dr. Ai is a fellow of The Institution of Engineering and Technology and an IEEE VTS Distinguished Lecturer. He received the Distinguished Youth Foundation and Excellent Youth Foundation from the National Natural Science Foundation of China, the Qiushi Outstanding Youth Award by the Hong Kong Qiushi Foundation, the New Century Talents by the Chinese Ministry of Education, the Zhan Tianyou Railway Science and Technology Award by the Chinese Ministry of Railways, and the Science and Technology New Star by the Beijing Municipal Science and Technology Commission. He is an IEEE VTS Beijing Chapter Vice Chair and an IEEE BTS Xi'an Chapter Chair. He was a Co-Chair or a Session/Track Chair of many international conferences. He is an Associate Editor of the IEEE ANTENNAS AND WIRELESS PROPAGATION LETTERS, the IEEE TRANSACTIONS ON CONSUMER ELECTRONICS, and an Editorial Committee Member of the \emph{Wireless Personal Communications Journal}. He is the Lead Guest Editor of special issues on the IEEE TRANSACTIONS ON VEHICULAR TECHNOLOGY, the IEEE ANTENNAS AND PROPAGATIONS LETTERS, and the \emph{International Journal on Antennas and Propagations}.
\end{IEEEbiography}

\begin{IEEEbiography}[{\includegraphics[width=1in,height=1.25in,clip,keepaspectratio]{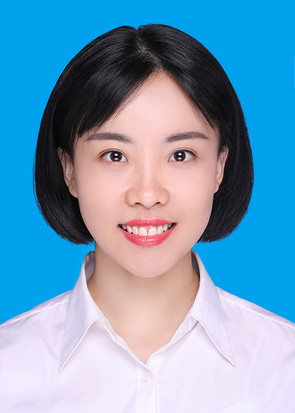}}]{Meilin Gao}(Member, IEEE) received the B.E. degree in Information Engineering from Shijiazhuang Tiedao University, Shijiazhuang, China, in 2013, and the Ph.D. degree at the State Key Laboratory of Rail Traffic Control and Safety, Beijing Jiaotong University, Beijing, China, in 2021. From August 2018 to January 2020, she worked as a Visiting Ph.D. Student with the BBCR Group, Department of Electrical and Computer Engineering, University of Waterloo, Canada. She is currently a postdoctoral fellow with the Tsinghua Space Center, Tsinghua University. Her research interests include millimeter-wave communications, wireless resource allocation, mobile edge caching and satellite communications.
\vspace{50mm}
\end{IEEEbiography}

\begin{IEEEbiography}[{\includegraphics[width=1in,height=1.25in,clip,keepaspectratio]{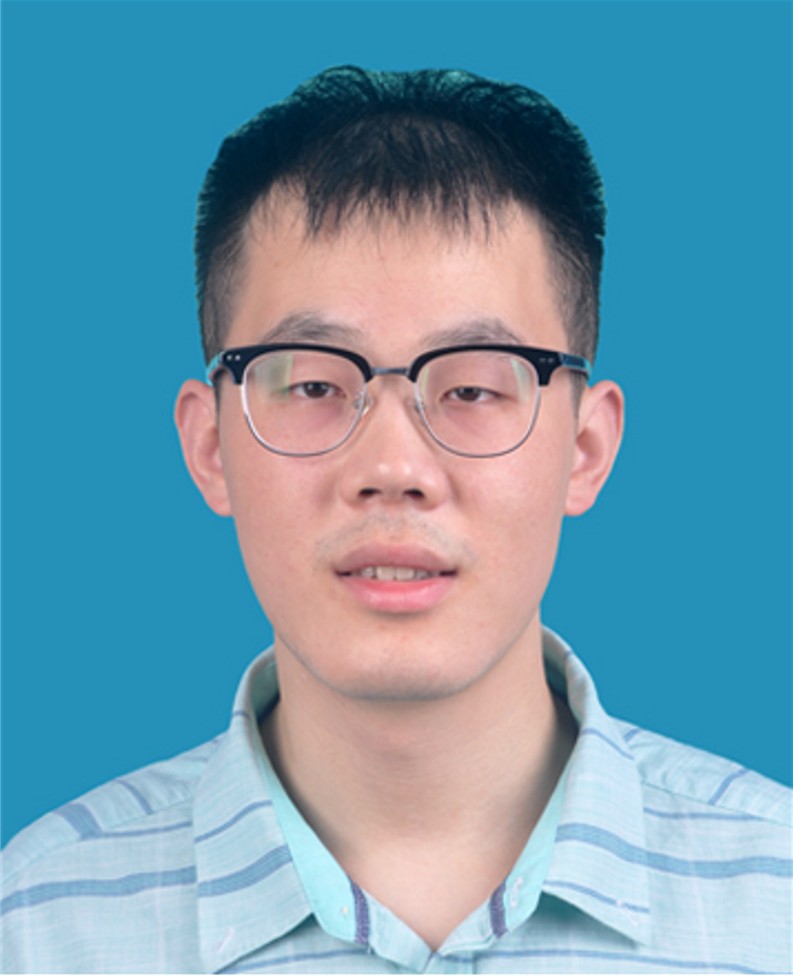}}]{Zhangfeng Ma} (Member, IEEE) received the B.S. degree in communications engineering from the Huaihua University, Huaihua, China, in 2015, the M.S. degree in electronic and communication engineering from the Chongqing University of Posts and Telecommunications (CQUPT), Chongqing, China, in 2018, and the Ph.D. degree (Hons.) in information and communication engineering from the School of Electronic and Information Engineering, Beijing Jiaotong University (BJTU), Beijing, China, in 2022.

From 2021 to 2022, he was a Visiting Ph.D. Student with the Antennas, Propagation and Millimetre-wave Systems (APMS) section, Department of Electronic Systems, Aalborg University, Aalborg, Denmark. He is currently an Associate Professor with the School of Information Science and Engineering, Shaoyang University, Shaoyang, China. His current research interests include wireless propagation channel modeling, non-stationary channel models, and unmanned aerial vehicle channel modeling. He was a recipient of the Best Student Paper Award from the IEEE Vehicular Technology Conference (VTC-Fall) in 2019, and the Outstanding Doctor's Thesis Award from the Chinese Institute of Electronics (CIE) in 2023.
\end{IEEEbiography}

\begin{IEEEbiography}[{\includegraphics[width=1in,height=1.25in,clip,keepaspectratio]{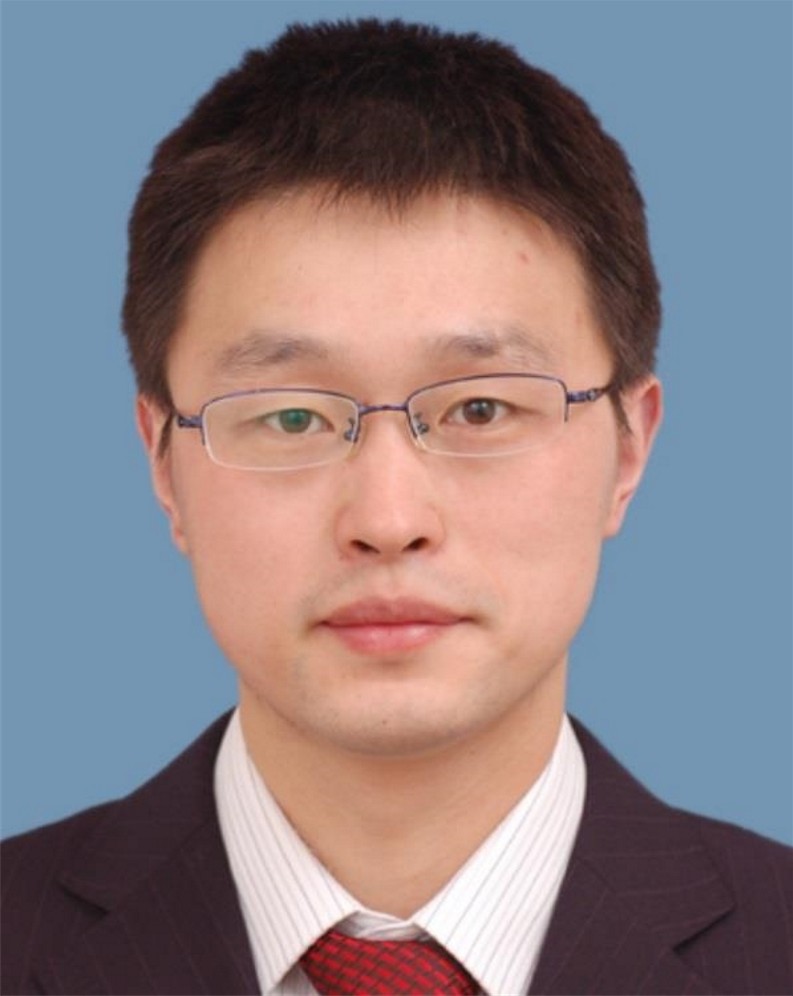}}]{Gongpu Wang}(Member, IEEE) received the B.Eng. degree in communication engineering from Anhui University, Hefei, Anhui, China, in 2001, the M.Sc.degree from the Beijing University of Posts and Telecommunications, Beijing, China, in 2004, and the Ph.D. degree from University of Alberta, Edmonton, Canada, in 2011. From 2004 to 2007, he was an Assistant Professor at the School of Network Education, Beijing University of Posts and Telecommunications. He is currently a Full Professor with the School of Computer and Information Technology, Beijing Jiaotong University, China. His research interests include wireless communication theory, signal processing technologies, and the Internet of Things.
\end{IEEEbiography}

\begin{IEEEbiography}[{\includegraphics[width=1in,height=1.25in,clip,keepaspectratio]{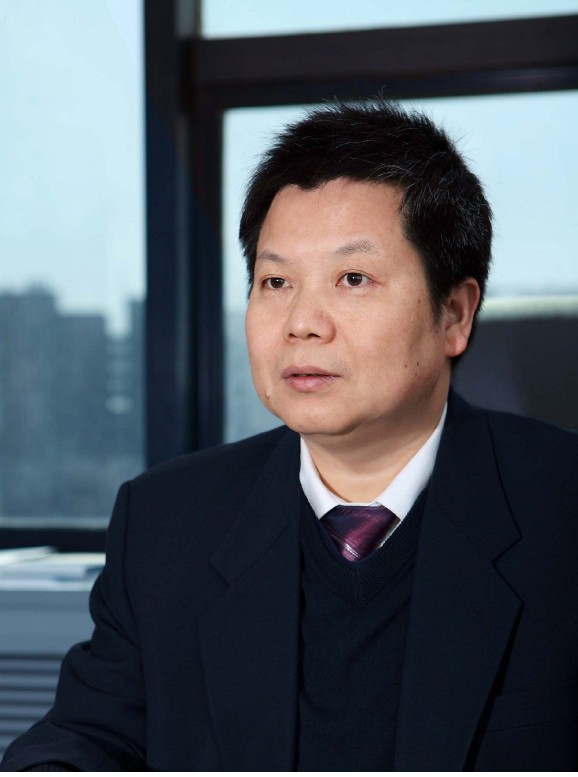}}]{Zhangdui Zhong}(Fellow, IEEE) received the B.E. and M.S. degrees from Beijing Jiaotong University, Beijing, China, in 1983 and 1988, respectively. He is a Professor and Advisor of Ph.D. candidates with Beijing Jiaotong University, Beijing, China. He is currently a Director of the School of Computer and Information Technology and a Chief Scientist of State Key Laboratory of Rail Traffic Control and Safety, Beijing Jiaotong University. He is also a Director of the Innovative Research Team of Ministry of Education, Beijing, and a Chief Scientist of Ministry of Railways, Beijing. He is an Executive Council Member of Radio Association of China, Beijing, and a Deputy Director of Radio Association, Beijing. His interests include wireless communications for railways, control theory and techniques for railways, and GSM-R systems. His research has been widely used in railway engineering, such as Qinghai-Xizang railway, Datong-Qinhuangdao Heavy Haul railway, and many high-speed railway lines in China.
  
He has authored or coauthored seven books, five invention patents, and over 200 scientific research papers in his research area. He received the MaoYiSheng Scientific Award of China, ZhanTianYou Railway Honorary Award of China, and Top 10 Science/Technology Achievements Award of Chinese Universities.
\end{IEEEbiography}

\end{document}